\documentclass[11pt,reqno]{amsart}

\usepackage[margin=1.2in]{geometry}
\usepackage{graphicx}
\usepackage{subcaption}
\usepackage{stmaryrd}
\usepackage{amssymb}
\usepackage{amsthm}
\usepackage{dsfont}
\usepackage{hyperref}
\usepackage{enumerate}
\usepackage{mathrsfs}
\usepackage{amsmath, dsfont, amssymb, amsthm, mathrsfs}
\usepackage{stmaryrd}
\usepackage[lite,initials,msc-links]{amsrefs}
\usepackage{amsmath}
\usepackage{centernot}
\usepackage{mathtools}
\usepackage{xcolor}
\usepackage[normalem]{ulem}

\usepackage{amsaddr}

\usepackage{comment}
\excludecomment{comment}

\usepackage{hyperref}
\hypersetup{
	colorlinks=true,
	linkcolor= black,
	citecolor = black,
	linktocpage = True,
}

\newtheorem{theorem}{Theorem}
\newtheorem{lemma}[theorem]{Lemma}
\newtheorem{proposition}[theorem]{Proposition}
\newtheorem{corollary}[theorem]{Corollary}
\numberwithin{equation}{section}
\theoremstyle{definition}

\usepackage{chngcntr}
\counterwithout{equation}{section}

\theoremstyle{definition}
\newtheorem*{definition}{Definition}
\theoremstyle{remark}
\newtheorem{remark}{Remark} 
\theoremstyle{question}
 
\allowdisplaybreaks

\newcommand{\n}{\mathbf{n}}

\newcommand{\bl}{\mathbf{b}}
\newcommand{\re}{\mathbf{r}}

\newcommand{\DCWire}{{\mathcal{L}}}

\newcommand{\dcwire}{ \omega}

\newcommand{\nat}{\mathbb N}

\newcommand{\ZZ}{\mathbb{Z}}
\newcommand{\PP}{\mathbb{P}}
\newcommand{\E}{\mathbb{E}}
\newcommand{\id}{\mathds{1}}
\newcommand{\NN}{\mathbb{N}}
\newcommand{\RR}{\mathbb{R}}

\renewcommand{\c}[1]{\mathcal{#1}}

\makeatletter
\newcommand\connect[2][]{%
  \ext@arrow 9999{\longleftrightarrowfill@}{#1}{#2}}
\newcommand\longleftrightarrowfill@{%
  \arrowfill@\leftarrow\relbar\rightarrow}
\makeatother

\title[An elementary proof of phase transition in the XY model]{An elementary proof of phase transition\\ in the planar XY model}

\author[Diederik van Engelenburg \and Marcin Lis]{ \small{Diederik van Engelenburg \and Marcin Lis}}

\address{University of Vienna}

\begin{document}

\maketitle

\begin{abstract} 
Using elementary methods we obtain a power-law lower bound on the two-point function of the planar XY spin model at low temperatures. 
This was famously first rigorously obtained by Fr\"{o}hlich and Spencer~\cite{FroSpe} and establishes a Berezinskii--Kosterlitz--Thouless phase transition in the model.
Our argument relies on a new loop representation of spin correlations, a recent result of Lammers~\cite{Lammers} on delocalisation of general integer-valued height functions, and classical correlation inequalities.
\end{abstract}

\section{Introduction and main result}

Let $G=(V,E)$ be a finite graph.
Given a collection of nonnegative \emph{coupling constants} $J=(J_e)_{e\in E}$, and an \emph{inverse temperature} $\beta>0$, the \emph{XY model} (with free boundary conditions) is a random spin configuration $\sigma\in \mathbb S^V$, where $\mathbb S=\{ z\in \mathbb C: |z|=1\}$ is the complex unit circle, 
sampled according to the Gibbs distribution
\begin{align} \label{def:xy}
d \mu_{G,\beta}(\sigma) \propto \exp \Big(\tfrac 12\beta\sum_{vv'\in E}J_{vv'}( \sigma_v\bar\sigma_{v'}+ \bar\sigma_v\sigma_{v'}) \Big) \prod_{v\in V} d \sigma_v,
\end{align}
where $vv'$ denotes the edge $\{v,v'\}$, and $d\sigma_v$ is the uniform probability measure on $\mathbb S$. 
For simplicity of notation, unless stated otherwise, we will assume that $J_e = 1$ for all $e$. However, our results extend naturally to nonhomogeneous coupling constants.
We will write $\langle \cdot \rangle_{G,\beta}$ for the expectation with respect to $\mu_{G,\beta}$. 
The observable of main interest for us will be the \emph{two-point function} $\langle \sigma_a\bar\sigma_b\rangle_{G,\beta}$, $a,b\in V$, and its \emph{infinite volume limit}
\[
\langle \sigma_a\bar\sigma_b\rangle_{\Gamma,\beta}= \lim_{G\nearrow \Gamma}\langle \sigma_a\bar\sigma_b\rangle_{G,\beta},
\] 
where $\Gamma$ is an infinite planar lattice.

Note that if $\sigma_v=e^{i \theta_v}$, $\theta_v\in (-\pi,\pi]$, then $\sigma_v\bar\sigma_{v'}+ \bar\sigma_v\sigma_{v'}=2\cos(\theta_v-\theta_{v'})$. This means that the model is \emph{ferromagnetic}, i.e., pairs of neighbouring spins that are (almost) aligned have smaller energy and hence are statistically favoured.
A natural question is wether varying $\beta$ leads to a ferromagnetic \emph{order--disorder} phase transition in the model. The classical theorem of Mermin and Wagner~\cite{MerWag} 
excludes this possibility when the underlying lattice $\Gamma$ is two-dimensional. Moreover, McBryan and Spencer showed that at any finite temperature $\langle \sigma_a\bar\sigma_b\rangle_{\mathbb Z^2,\beta}$ decays to zero at least as fast as a power of the distance between $a$ and $b$. On the other hand, it is known by the work of Fr\"{o}hlich, Simon and Spencer~\cite{FSS} that in higher dimensions the model exhibits long-range order at low temperatures and the two-point function does not decay to zero.

Even though there is no spontaneous symmetry breaking, Berezinskii~\cite{Ber1,Ber2}, and Kosterlitz and Thouless~\cite{KosTho} predicted that a different type of phase transition takes place in two dimensions.
It should be understood in terms of interacting topological excitations of the model, the so called \emph{vortices} and \emph{antivortices}. They are those faces of the graph where the XY configuration makes a full clockwise or anticlockwise turn respectively when one traverses the edges of the face in a clockwise manner. Vortices and antivortices interact through a Coulomb interaction, and are energetically favoured 
to form short-distance pairs of vortex-antivortex. However, such configurations have clearly much smaller entropy. The \emph{Berezinskii--Kosterlitz--Thouless} (BKT) phase transition happens when, while increasing the temperature, entropy wins agains energy, and the vortex-antivortex pairs unbind and form a plasma of  
freely spaced vortices and antivortices. This regime corresponds to exponential decay, whereas the phase with bound vortex-antivortex pairs should exhibit power-law decay of the two-point function. 
A rigorous lower bound of this type for low temperatures, and therefore a proof of the BKT phase transition was first obtained in the celebrated work of Fr\"{o}hlich and Spencer~\cite{FroSpe} who also derived analogous results for the Villain spin model. Their proof uses a multi-scale analysis of the Coulomb gas, and the main purpose of the present article is to present an alternative and less technically involved argument for the existence of phase transition in two dimensions.

To be more precise, we introduce a new loop representation for the two-point function in the XY model that can be used to transfer probabilistic information from the dual integer-valued height function model to the XY model. Along the way we also show that the height function possesses the crucial absolute-value-FKG property.
This, together with a recent elementary delocalisation result for general height functions obtained by Lammers~\cite{Lammers}, is used to prove existence of the BKT phase transition.

\begin{theorem}[Berezinskii--Kosterlitz--Thouless phase transition] \label{T:main_BKT}
	There exists $\beta_c\in (0,\infty)$ such that 
		\begin{enumerate}[(i)]
		\item for all $\beta < \beta_c$, there exists $c =c(\beta)> 0$ such that for all $v,v'\in \mathbb Z^2$,
		\[
		\langle \sigma_v \overline{\sigma}_{v'} \rangle_{\mathbb Z^2,\beta} \leq e^{-c|v - v'|},
		\]
		\item for all $\beta \geq \beta_c$ and all  distinct $v,v'\in \mathbb Z^2$, 
		\[
		\langle \sigma_v \overline{\sigma}_{v'} \rangle_{\mathbb Z^2,\beta} \geq \frac{1}{8|v - v'|}.
		\]
	\end{enumerate}
\end{theorem}
We note that unlike in the original proof of Fr\"{o}hlich and Spencer, we do not show that the rate of decay approaches zero when so does the temperature.
However, we establish a type of sharpness which says that there is no other behaviour than exponential and power-law decay.
The short proof of sharpness is independent of the rest of the argument. In the first step we classically use the Lieb--Rivasseau inequality~\cite{Lieb,Riv} to establish a sharp transition between exponential decay and nonsummability of correlations (similarly to the proof for the Ising model~\cite{DCT}). To conclude a uniform power-law lower bound as in $(ii)$ whenever the correlations are not summable we use the 
Messager--Miracle-Sole inequality~\cite{MMS} on monotonicity of correlations with respect to the position of the vertex on the lattice.

We also note that our proof works (with minor modifications and a different, implicit multiplicative constant in $(ii)$) for other infinite graphs that in addition to being translation invariant possess reflection and rotation symmetries.

For a more detailed overview of the XY model, we refer the reader to~\cite{PelSpi, FriVel}, and for expositions of the argument of Fr\"{o}hlich and Spencer, we refer to~\cite{KP,GS}.

This article is organised as follows. 
\begin{itemize}
\item In Section~\ref{sec:HF} we introduce the dual of the planar XY model in form of an integer-valued height function defined on the faces of the graph. 
We also establish positive association of its absolute value (the absolute-value-FKG property), and recall the delocalisation result of Lammers~\cite{Lammers}.
\item
In Section~\ref{sec:singleswitch} we define a random collection of loops on the graph that carries probabilistic information about both the XY spins and the dual height function.
Although this is a well known object that goes back to the works of Symanzik~\cite{Sym}, and Brydges, Fr\"{o}hlich and Spencer~\cite{BFS}, the formula that relates the two-point function 
to the probability of two points being connected by a loop~(Lemma~\ref{L:SingleSwitching}) is new and crucial to our argument. 
\item In Section~\ref{sec:noexpdec} we give an elementary argument which states that if the height function delocalises at some temperature, then the spin two-point function \emph{does not} decay exponentially.
\item In Section~\ref{sec:existence} we use the above ingredients to show that on any translation invariant graph, there exists a finite temperature at which the two-point function does not decay exponentially.
This is not immediate as the result of Lammers~\cite{Lammers} applies only to trivalent graphs. However, a simple graph-modification argument together with the Ginibre inequality allows to 
change the setup from a general graph to a triangulation (a graph whose dual is trivalent).
\item In Section~\ref{sec:Liebsharp} we finish the proof of the main theorem. We use the Lieb--Rivasseau inequality~\cite{Lieb,Riv} and the inequality of Lemma~\ref{L: ferro-magnet inequal} to show that the absence of exponential decay implies 
a power-law lower bound on the two-point function. The proof of sharpness relies only on this section and Lemma~\ref{L: ferro-magnet inequal}.
\item Finally, independently of the rest of the article, in Appendix~\ref{sec:doubleswitch} we develop a new loop representation for \emph{squares and products} of spin correlation functions in the XY model. 
As an application we present new correlation inequalities and give new combinatorial proofs of the Lieb--Rivasseau~\cite{Lieb,Riv}, and the Messager--Miracle-Sole inequality~\cite{MMS}.
We hope this representation will be useful in the further study of the XY model.
\end{itemize}

\subsection*{Acknowledgements} ML is grateful to Roland Bauerschmidt and Hugo Duminil-Copin for inspiring discussions on the XY model. We also thank Christophe Garban for useful remarks on a draft.

\section{The dual height function} \label{sec:HF}

To define the dual model we assume that $G$ is planar and we need to introduce currents. To this end, let $\vec E$ be the set of directed edges of $G$, and let $\nat =\{0,1,\ldots\}$.
A function $\n: \vec E\to \mathbb N$ is called a \emph{current} on $G$. For a current~$\n$, we define $\delta \n: V\to \mathbb Z$ by 
\[
\delta\n_v= \sum_{v'\sim v} \n_{(v,v')} -   \n_{(v',v)}.
\]
Hence if $\delta\n_v$ is positive, then the amount of outgoing current is larger than the incoming current, an we think of $v$ as a \emph{source}.
Likewise if $\delta\n_v$ is negative, there is more incoming current and $v$ is a \emph{sink}.
A current is \emph{sourceless} if $\delta\n_v=0$ for all $v\in V$.

We define $\Omega_0$ to be the set of all (sourceless) currents. 
Sourceless currents naturally define a height function $h$ on the set of faces of $G$, denoted by $U$, where the height of the outer face is set to zero, and the increment of the height 
between two faces $u$ and $u'$ is equal to 
\begin{align*}
h(u)-h(u')=\n_{(v,v')} -   \n_{(v',v)},
\end{align*} 
where the primal directed edge $(v,v')$ crosses the dual directed edge $(u,u')$ from right to left.
That this yields a well defined function on the faces of $G$ follows from the fact that $\delta \n=0$.
We define the \emph{XY weight} of a current by
\begin{align} \label{eq:XYweight}
w_{\beta}(\n)=\prod_{(v,v') \in \vec E}\frac{1}{\n_{(v,v')}!}\Big(\frac{\beta J_{vv'}}2 \Big)^{\n_{(v,v')}},
\end{align}
These weights appear naturally in the expansion of the partition function of the XY model into a sum over sourceless currents after one expands the exponentials in~\eqref{def:xy} into a power series in the variables $\tfrac12\beta J_{vv'}\sigma_v\bar\sigma_{v'}$ for each directed edge $(v,v')\in \vec E$,
and then integrates out the $\sigma$ variables. They will also appear in the analogous classical expansion for spin correlations~\eqref{eq:currentexp}.

We note that using currents to define a model on the dual graph is an instance of planar duality of
abelian spin systems~\cite{dubedat2011topics}, and the fact that the function is integer valued is a consequence of $\mathbb Z$ being the dual group of the unit circle.

Clearly, the weight~\eqref{eq:XYweight} defines a probability measure $\PP_{G,\beta}$ on currents and hence also on height functions. In terms of the height function it is a Gibbs measure given by
\begin{align} \label{eq:Gibbs}
\PP_{G,\beta}(h) \propto \exp\Big(- \sum_{uu'\in E^\dagger} \mathcal V_e^{\beta}(h(u)-h(u'))\Big),
\end{align}
where $E^\dagger$ is the set of dual edges of $G$, and where the symmetric potentials $ \c{V}^{\beta}_e: \mathbb Z\to \mathbb R$ are given by
\begin{align} \label{eq:potentialBessel}
\mathcal V^{\beta}_e(k) =- \log \Big(\sum_{i= 0}^{\infty} \frac{1}{i!(i+|k|)!}\Big(\frac{\beta J_e}2\Big)^{2i+|k|} \Big) =- \log I_{k}({\beta J_e })
\end{align}
with $I_{k}$ being the modified Bessel function.
We again note that we will usually set all $J_e=1$ to simplify the notation.

A well known Tur\'{a}n-type inequality for modified Bessel functions~\cite{Bessel} states that for any $k\geq 0$ and $\beta >0$, 
\begin{equation} \label{E: Turan}
I_k^2(\beta)\geq I_{k-1}(\beta) I_{k+1}(\beta)
\end{equation} which means that $\mathcal V^\beta_e$ is convex on the integers.
This puts the model in the well studied framework of height functions with a convex potential (see e.g.\ \cite{sheffield}).

\subsection{Gibbs measures and delocalisation}
To state the delocalisation result of Lammers we will need the notion of a Gibbs measure for height functions on infinite graphs (though we will not directly work with it in the remainder of the article). Let $\Gamma = (V, E)$ be an infinite planar graph and $\Gamma^{\dagger} = (U, E^{\dagger})$ its planar dual. If $\nu$ is a measure on height functions $\varphi: \ZZ^{U} \to \ZZ$ and $\Lambda \subset U$ a finite subset, write $\nu_{\Lambda}$ for the measure restricted to $\Lambda$. Let $\mathcal V=(\mathcal V_e)_{e\in E\dagger}$ be a family of convex symmetric potentials. We call $\nu$ 
a \emph{Gibbs measure} for the potential $\mathcal V$ if for every such $\Lambda$, it satisfies the Dobrushin--Lanford--Ruelle relation
\[
	\nu_{\Lambda}(\cdot)= \int_{\ZZ^{U}}  \nu^{\varphi}_{\Lambda}(\cdot) d\nu(\varphi),
\]
where $\nu^{\varphi}_{\Lambda}$ is the Gibbs measure on height functions $h\in \ZZ^U$ given as in~\eqref{eq:Gibbs} (but with $\mathcal V^\beta$ replaced by $\mathcal V$) and conditioned on $h$ being equal to $\varphi$ on the boundary of $\Lambda$.

In what follows we will always assume that $\Gamma$ is locally finite and invariant under the action of a $\mathbb Z^2$-isomorphic lattice.
We say that $\nu$ is translation invariant if it is invariant under the same acton.

In a recent beautiful work~\cite{Lammers} Lammers gave a condition on the potential that guarantees that there are no translation invariant Gibbs measures on graphs of degree three (trivalent graphs).
\begin{theorem}[Lammers~\cite{Lammers}]\label{thm:lammers}
Let $\Gamma^{\dagger} = (U, E^{\dagger})$ be as above and moreover trivalent. If for every $e\in E^\dagger$,
\begin{equation} \label{E: exited potential}
	\mathcal V_e(\pm 1) \leq\mathcal  V_e(0) + \log(2),
\end{equation}
then there are no translation invariant Gibbs measures for $\mathcal V$.
\end{theorem}
This together with the dichotomy stated in Theorem~\ref{T: dichotomy} will be one of the key ingredients of the proof of the main theorem.

\subsection{Absolute-value-FKG and dichotomy}
In this section, we prove that the height function satisfies the absolute-value-FKG property, which is known to imply the following dichotomy.

Let $\Gamma = (V, E)$ be a translation invariant graph, and let $0$ be a chosen face of~$\Gamma$. 
Define $B_n$ to be the subgraph of $\Gamma$ induced by the vertices in $V$ that lie on at least one face of $\Gamma$ that is contained in the graph ball of radius $n$ on $\Gamma^\dagger$.
We introduce this slightly convoluted definition to guarantee the following three properties: $0$ belongs to all $B_n$, also $B_n\nearrow \Gamma$ as $n\to \infty$, and finally, the weak dual graph of $B_n$ (the dual graph with the vertex corresponding to the external face of $B_n$ removed) is a subgraph of $\Gamma^\dagger$.
\begin{theorem} \label{T: dichotomy}
	Consider the setup as above. Then
	for every $\beta > 0$, exactly one of the following two occurs:
	\begin{enumerate}[(i)]
		\item (localisation) There exists a $C < \infty$ such that uniformly over all $n$, 
		\[
			\E_{B_n,\beta}[ |h(0)|] \leq C.
		\]
		\item (delocalisation) There are no translation invariant Gibbs measures for the potential~\eqref{eq:potentialBessel}. 
	\end{enumerate}
\end{theorem} 

\begin{proof}
	This is a consequence of the absolute-value-FKG property proved below (Proposition~\ref{P: abs-FKG}) and standard arguments using monotonicity in boundary conditions. See for example \cite[Lemma 2.2]{CPST} or \cite[Theorem 2.7]{LamOtt}.
\end{proof}

The remainder of this section is devoted to proving the following version of the absolute-value-FKG property.

\begin{proposition} \label{P: abs-FKG}
	Let $G = (V, E)$ be a finite graph and $U$ the set of its faces. Then for all $\beta>0$, and all $\Psi, \Phi: \nat^{U} \to \RR_+$ increasing functions, 
	\[
		\E_{G,\beta}[\Psi(|h|)\Phi(|h|)] \geq \E_{G,\beta}[\Psi(|h|)] \E_{G,\beta}[\Phi(|h|)]. 
	\]
\end{proposition}

The proposition is easiest to prove for small $\beta$. We extend this to general $\beta$ afterwards.

\begin{lemma} \label{L: absFKG beta small}
	The above holds true for all $\beta \leq 1$. 
\end{lemma}

\begin{proof}
	We rely on a result of Lammers and Ott \cite[Theorem 2.8]{LamOtt}, stating that if
	\begin{align*}
		\mathcal V_e^{\beta}(k - 1) - 2\mathcal V_e^{\beta}(k) + \mathcal V_e^{\beta}(k + 1) = -\log\Big(\frac{I_{k-1}(\beta)I_{k + 1}(\beta)}{I_k(\beta)^2}\Big)
	\end{align*}
	is a nonincreasing function of $k$ on $\{0,1,\ldots\}$, then $\mathbb P_{G,\beta}$ is absolute-value-FKG. We define $r_k = \tfrac{1}{\beta}\frac{I_{k }(\beta)}{I_{k-1}(\beta)}$, 
	and need to show that $r_k^2\leq r_{k-1}r_{k+1}$ for all $k\geq 0$.
	The well known recurrence relation 
	\[
	I_{k - 1}(\beta) = \tfrac{2k}{\beta} I_{k}(\beta) + I_{k + 1}(\beta) \qquad \text{yields} \qquad r_k = (2k + \beta^2r_{k + 1})^{-1}.
	\]
	Hence it is enough to prove that
	\[
		(2k + \epsilon_{k + 1})(2(k + 2) + \epsilon_{k + 3})\leq (2(k + 1) + \epsilon_{k + 2})^2,
	\]
	where $\epsilon_k = \beta^2r_{k}$. Using the Tur\'{a}n inequality \eqref{E: Turan}, it follows that $0 \leq r_{k + 1} \leq r_{k}$, and therefore it is sufficient to establish that
\[
		R_k:=(2k + \epsilon_{k + 1})(2k + 4 + \epsilon_{k + 1}) - (2k + 2)^2=4(k+1)\epsilon_{k + 1} + \epsilon_{k + 1}^2  - 4 \leq 0.
	\]
	At the same time, simply using the definition of $r_{k+1}$ and comparing the Taylor expansions~\eqref{eq:potentialBessel} of $I_{k+1}$ and $I_k$ term by term gives $\epsilon_{k+1} \leq \beta^2/(2k+2)$. Therefore, when $\beta\leq 1$, we have $R_k\leq \epsilon_{k + 1}^2  - 2 \leq 0$ for all $k\geq0$, which concludes the proof. 
\end{proof}

To treat general values of $\beta$, we will use a trick which consists in replacing each edge of $G$ by $s=\lceil\beta \rceil $ consecutive edges, and reducing the parameter $\beta$ by the factor $s$, 
together with the following convolution property of the modified Bessel functions.

\begin{lemma} \label{L: convolution of I}
	For all $k, l \in \ZZ$ and all $\beta,\beta'\geq0$, 
	\[
		\sum_{m \in \ZZ} I_{k - m}(\beta)I_{m - l}(\beta') = I_{k - l}(\beta+\beta'). 
	\]
\end{lemma}
\begin{proof}
This is a classical identity which follows from the fact that $I_{k }(\beta)/e^{\beta}=\mathbb P(Z-Z'=k)$, where $Z,Z'$ are independent Poisson random variables with mean $\beta/2$, and 
the fact that a sum of independent Poisson random variables is Poisson.
\end{proof}

With this we can prove Proposition \ref{P: abs-FKG}.

\begin{comment}
\begin{proof}
	Fix $\beta > 0$ and $k, l \in \ZZ$. Straightforward computations show
	\begin{align*}
		&\sum_{m \in \ZZ} I_{k - m}(\beta / 2)I_{m - l}(\beta / 2) \\
		&= \sum_{n_1, n_2 \geq 0}^\infty \left(\frac{\beta}{4}\right)^{2(n_1 + n_2) + k - l} \sum_{m \in \ZZ} \frac{1}{n_1!n_2!(n_1 + k - m)!(n_2 + m - l)!} \\
		&= \sum_{n \geq 0} \left(\frac{\beta}{4}\right)^{2n + (k - l)} \sum_{n_1 + n_2 = n} \frac{1}{n_1!n_2!} \sum_{m = 0}^{n + k - l} \frac{1}{m!(n + k - l - m)!} \\
		&= \sum_{n \geq 0} \left(\frac{\beta}{4}\right)^{2n + (k - l)} \frac{2^n}{n!} \frac{2^{n + k - l}}{(n + k - l)!} = I_{k - l}. 
	\end{align*}
\end{proof}
\end{comment}

\begin{proof}[Proof of Proposition \ref{P: abs-FKG}]
	Let $G_s = (V_s, E_s)$ be $G$ with each edge replaced by $s$ consecutive edges, and let $h_s$ be the height function on $G_s$ with law $\mu_{G_s, \beta/s}$.
	By Lemma \ref{L: convolution of I} (and an induction argument) the restriction of $h_s$ to $V$ has the same law as $h_1$. Moreover, $\beta /s \leq 1$ by definition of $s$, which by Lemma \ref{L: absFKG beta small} implies that $\mu_s$ is absolute-value-FKG. To finish the proof it is enough to notice that any increasing function on $\mathbb N^V$ is also increasing on $\mathbb N^{V_s}$.
	\end{proof}

\begin{remark}
An interesting consequence of the idea above (that we will not use in this article) is the following. Consider the case when $s$ from above is independent of $\beta$ and diverges to infinity. 
In this limit, the height function becomes well defined at every point of every dual edge. Here we think of the dual graph as the so called \emph{cable graph}, i.e., every 
dual
edge $e$ is identified with a continuum interval of length $J_e\beta$. Then the distribution of the height on an edge, when conditioned on the values at the endpoints, is one of the difference of two Poisson processes with intensity $J_e\beta/2$ each,
and conditioned on the value at the endpoints. One can check that the model exhibits a spatial Markov property on the full cable graph and not only on the vertices.
This is in direct analogy with the cable graph representation of the discrete Gaussian free field, where the vertex-field can be extended to the edges via Brownian bridges (see e.g.~\cite{Lupu} and the references therein).
\end{remark}

\section{Loop representation of currents and path reversal} \label{sec:singleswitch}
The purpose of this section is mainly to develop a loop representation for the two-point function of the XY model.
The important aspect of our approach is that the correlations are represented as probabilities for loop connectivities in random ensembles of closed loops.
This is in contrast with most of the classical representations that write correlation functions as ratios of partition functions of loops, where in the numerator, in addition to loops, one
also sums over open paths between the points of insertion in the correlator~\cite{Sym,BFS}.
We note that a similar idea to ours appears in the work of Benassi and Ueltschi~\cite{BenUel}, but due to technical differences in the framework (see Remark~\ref{rem:BenUel}), the formula for the two-point function obtained in~\cite{BenUel} is not as transparent as ours.

Let $G=(V,E)$ be a finite, not necessarily planar graph. 
We say that a multigraph $\mathcal M$ on $V$ is a \emph{submultigraph} of $G$ if after identifying the multiple copies of the same edge in $\mathcal M$ it is a subgraph of $G$.

\begin{definition}[Loop configurations outside $S$]
Let $\mathcal M$ be a submultigraph of $G$, and let $S\subseteq V$. A \emph{loop configuration (on $\mathcal M$) outside $S$} is 
a collection of  
\begin{itemize} 
\item unrooted directed loops on $\mathcal M$ avoiding $S$, and
\item directed open paths on $ \mathcal M$ starting and ending in $S$ (and not visiting $S$ except at their start and end vertex),
\end{itemize}
such that every edge of $\mathcal M$ is traversed exactly once by a loop or a path. 

We write $\DCWire^S$ for the set of all loop configurations outside $S$, and define a weight for $\omega\in \DCWire^S$ by
\begin{equation} \label{eq: colour weight}
	 \lambda^S_\beta(\dcwire) = \prod_{v \in V \setminus S} \frac{1}{(\deg_{\mathcal M}(v) / 2)!}  \prod_{e\in E} \frac{1}{\mathcal M_e!} \Big(\frac{\beta}{2}\Big)^{\mathcal M_e},
\end{equation}
where $\mathcal M$ is the underlying multigraph, and $\mathcal M_e$ is the number of copies of $e$ in $\mathcal M$.  
When $S=\emptyset$, a configuration is composed only of loops 
that can visit every vertex in $V$, and we simply call it a loop configuration. 
\end{definition}
An important feature of the weight~\eqref{eq: colour weight} is that it depends on $\dcwire$ only through $\mathcal M$. 
Also note, that if $S'\subseteq S$, then there is a natural map $\rho: \mathcal L^{S'} \to \mathcal L^S$ that consists in forgetting (or cutting) the loop connections at the vertices in $S\setminus S'$.
Under this map, each configuration in $ \mathcal L^S$ has $\prod_{v \in S \setminus S'} {(\deg_{\mathcal M}(v) / 2)!}$ preimages, each of them having the same weight, and hence
\begin{align}\label{eq:cutting}
\sum_{\tilde \omega \in \rho^{-1}[\omega]} \lambda^{S'}_\beta(\tilde \dcwire) =  \lambda^S_\beta(\dcwire).
\end{align}  
This consistency property will be useful later on.

For now, let $|\n|: E\to \mathbb N$ be the \emph{amplitude} of a current $\n$, i.e.
\[
|\n|_{vv'}:=\n_{(v,v')}+\n_{(v',v)}.
\]

\begin{definition}[Multigraph of a current and consistent configurations]
For a current $\n$, let $\mathcal M_{\n}$ be the submultigraph of $G$ where each edge $e\in E$ is replaced by $|\n|_e$ (possibly zero) parallel copies of $e$.
A loop configuration on $\mathcal M_{\n}$ is called \emph{consistent with $\n$} if for every edge $(v,v')\in\vec E$, the number of times the loops traverse a copy of $vv'$
in the direction of ${(v,v')}$ is equal to 
$\n_{(v,v')}$.
We define $ {\DCWire}^S_{\n}$ to be the set of all loop configurations on $\mathcal M_{\n}$ outside $S$ that are consistent with $\n$. 
\end{definition}

For $\varphi: V\to \mathbb Z$, let $\Omega_\varphi=\{\n : \delta \n =\varphi \}$,
\[
Z^{\varphi}_{G,\beta} =\sum_{\n \in \Omega_\varphi} w_{\beta}(\n),
\]
and $\mathcal S (\varphi)=\{ v\in V: \varphi_v\neq 0\}$.
For a current $\n$, with a slight abuse of notation, we also write $\mathcal S(\n)=\mathcal S(\delta \n)$.
Note that $ {\DCWire}^S_{\n}$ can be nonempty only if $\mathcal S(\n)\subseteq S$. Indeed, each path and loop that enters a vertex in $V\setminus S$ must also leave it,
and hence the total number of incoming and outgoing arrows at each such vertex must be the same.
For $\varphi: V\to\mathbb Z$, we also define
\[
  \DCWire_{\varphi}^S=\bigcup_{\n\in \Omega_{\varphi}}  \DCWire^S_{\n}.
\] 
Again, this is nonempty only if $\mathcal S(\varphi)\subseteq S$.
We will write $  \DCWire_{0}^S$, where $0$ denotes the zero function on $V$.

We now relate the weights of loops to those of currents. To this end, note that for each edge $vv'\in E$, there are exactly 
\[
	\frac{|\n|_{v v'}!}{\n_{(v, v')}!\n_{{(v', v)}}!}
\]
ways of assigning orientations to it so that the result is consistent with $\n$.
Moreover, independently of the choices of orientations, there are exactly $(\deg_{\mathcal M_{\n}}(v) / 2)!$ possible pairings of the incoming and outgoing edges at each vertex $v \in V\setminus S$.
Combining all this we arrive at a crucial loop representation for current weights: if $\mathcal S(\n)\subseteq S$, then	
\begin{align} \label{eq:loopexp}
w_\beta(\n)= \sum_{\dcwire \in  \DCWire^S_{\n}}  \lambda^S_\beta(\dcwire).
\end{align}
An important observation here is that the left-hand side is independent of $S$, and hence so is the right-hand side.

\subsection{Coupling with the height function}
We now apply this framework to the case of two sourceless currents and a coupling with the corresponding height function. From~\eqref{eq:loopexp} we have
\begin{align} \label{eq:squarepartition}
Z^0_{G,\beta}= \sum_{\dcwire \in  \DCWire^\emptyset_0}  \lambda^{\emptyset}_\beta(\dcwire) 
\end{align}
where $0$ denotes the zero function on $V$.
\begin{remark}
This loop representation of the partition function, though obtained via a different procedure, goes back to the work of Symanzik~\cite{Sym}, and Brydges, Fr\"{o}hlich and Spencer~\cite{BFS}.
\end{remark}

Moreover, in the case when $G$ is planar we immediately get the following distributional identity. 
Define $  {\mathbf{P}}_{G,\beta}$ to be the probability measure on $ {\mathcal L}_0:=\mathcal L^\emptyset_0$ induced by the weights $ \lambda_{\beta}:=\lambda^\emptyset_{\beta}$.
For each face $u\in U$ of $G$, and $\omega\in \mathcal L_0$, define $W_{\omega}(u)$ to be the total net winding
of all the loops in $\omega$ around $u$. 
\begin{proposition} \label{prop:netwinding}
The law of $(W(u))_{u\in U}$ under ${\mathbf{P}}_{G,\beta}$ is the same as the law of the height function $(h(u))_{u\in U}$ under ${\mathbb P}_{G,\beta}$.
\end{proposition}

\subsection{The two point-function and path reversal} \label{sec:twopoint}
We now turn to the loop representation of the two-point function. 
For reasons that will become apparent soon, we need to consider the two-point function of the squares, i.e., $\langle \sigma^2_a \bar \sigma^2_b \rangle$.
We note that the more standard correlation function $\langle \sigma_a \bar \sigma_b \rangle$ (or rather its square) can be treated using our approach from Appendix~\ref{sec:doubleswitch}.

Since the resulting currents will have sources, we will need to consider nonempty $S$ in the construction above.
To this end, fix two vertices $a,b\in V$, and
and define $\varphi =2(\delta_a-\delta_b)$, where $\delta_a(v)=\id\{ a=v\}$.
To lighten the notation, will write $a,b$ instead of $\{a,b\}$ for the set $S$. 
As for the partition function, expanding the exponential in the Gibbs--Boltzmann weights~\eqref{def:xy} into a power series in $\tfrac12 \beta J_{vv'}\sigma_v\bar \sigma_{v'}$ for each directed $(v,v')\in \vec E$, and integrating out the~$\sigma$ variables, we classically get 
\begin{align} \label{eq:currentexp}
\langle \sigma^2_a \bar \sigma^2_b \rangle_{G,\beta}= \frac{Z^{\varphi}_{G,\beta} }{Z^0_{G,\beta}}
= \frac{\sum_{ \dcwire \in  \DCWire^{a,b}_{\varphi}}  \lambda^{a,b}_\beta(\dcwire)}{Z^0_{G,\beta}},
\end{align}
where the last equality is new and follows from \eqref{eq:loopexp}.

We will write $\mathcal P_{a,b}(\dcwire)$ for the set of paths in $\omega$ that start at $a$ and end at $b$, and define
\[
m_{a,b}(\omega)=|\mathcal P_{a,b}(\dcwire)|.
\]
We now want to ``erase the sources'' at $a$ and $b$ from the currents underlying $ \DCWire^{a,b}_{\varphi}$, and hence rewrite the numerator as a sum over $ \DCWire^{a,b}_{0}$. We will then ultimately connect the open paths at $a$ and $b$ in all possible ways, and hence get a sum over $ \DCWire^\emptyset_{0}$ (see Figure \ref{f:path-switching} for an example). To this end note that in each $\dcwire\in  \DCWire^{a,b}_{\varphi}$ there are exactly two more paths going from $a$ to $b$, than those going from $b$ to $a$, i.e., $m_{a,b}(\omega)=m_{b,a}(\omega)+2$. 
The elementary operation that we will perform on the former paths is reversal. To this end, denote by $r(\gamma)$ the path $\gamma$ with the orientation of all the visited edges reversed.
Obviously this does not change the underlying multigraph, and hence also the weight of the loop configuration. 
The crucial observation now is that it maps $\dcwire\in  \DCWire^{a,b}_{\varphi}$ to a configuration $ \dcwire' \in  \DCWire^{a,b}_{0}$, and hence erases the sources of the underlying currents.
Indeed one can easily check that after reversing a path, the number of incoming minus the number of outgoing edges at every vertex $v\notin  \{a,b\}$ in $\dcwire'$ is the same as in $\dcwire$, whereas at $a$ (resp.\ $b$) this number is decreased (resp.\ increased) by two. 
More precisely, our transformation maps bijectively a pair $(\dcwire,\gamma)$ where $\dcwire\in \DCWire^{a,b}_{\varphi}$ and $\gamma \in \mathcal P_{a,b}(\omega) $
to the pair $( \dcwire', r(\gamma))$ where $ \dcwire' \in  \DCWire^{a,b}_0$ and $r( \gamma )\in \mathcal P_{b,a}(\omega')$.
Moreover, $m_{b,a}( \dcwire')= m_{b,a}(\dcwire)+1$, which in particular means that $m( \dcwire')> 0$.
Since path reversal does not change the weight of a loop configuration, we obtain
\begin{align*}
\sum_{   \dcwire \in \DCWire^{a,b}_{\varphi}} \lambda^{a,b}_\beta( \dcwire) &=\sum_{  \dcwire \in  \DCWire^{a,b}_{\varphi}, \gamma \in  \mathcal P_{a,b}( \omega)}\frac{1}{m_{b,a}( \omega)+2} \lambda^{a,b}_\beta( \dcwire) \\
&=\sum_{   \dcwire' \in  \DCWire^{a,b}_0,  \gamma' \in  \mathcal P_{b,a}(  \omega')}\frac{1}{ m_{b,a}( \dcwire')+1 } \lambda^{a,b}_\beta( \dcwire') \id \{ m_{a,b}( \dcwire') > 0 \} \\ 
&={\sum_{   \dcwire '\in  \DCWire^{a,b}_0}}\frac{m_{b,a}( \dcwire') }{m_{b,a}( \dcwire')+1 } \lambda^{a,b}_\beta( \dcwire') \id\{ m_{b,a}( \dcwire') > 0 \} \\
&={\sum_{   \dcwire'' \in  \DCWire^\emptyset_0}}\frac{m_{b,a}( \dcwire'') }{m_{b,a}( \dcwire'')+1 } \lambda^\emptyset_\beta( \dcwire'') \id\{ m_{b,a}( \dcwire'') > 0 \} ,
\end{align*}
where in the second equality we used path reversal, the last equality follows from \eqref{eq:cutting} with $S'=\emptyset$, and where, with a slight abuse of notation, for $\dcwire''\in \DCWire^\emptyset_0$, $m_{b,a}( \dcwire'') $ is the number of pieces of loops going from $b$ to $a$ and not visiting $b$ nor $a$ except for the start and end vertex.
Recall that $  {\mathbf{P}}_{G,\beta}$ is the probability measure on $ \mathcal L^\emptyset_0$ induced by the weights~$ \lambda^\emptyset_{\beta}$,
and note that $m_{b,a}$ has the same distribution as $m_{a,b}$  under $  {\mathbf{P}}_{G,\beta}$ (the law on loops is invariant under a global orientation reversal).
We therefore obtain from \eqref{eq:squarepartition} and \eqref{eq:currentexp} the following loop representation of the two-point function. 

\begin{lemma} \label{L:SingleSwitching}
	Let $a,b\in V$ be distinct. Then
	\[
		\langle \sigma^2_{a}\bar{\sigma}^2_{b} \rangle_{G,\beta} = {\mathbf{E}}_{G,\beta}\Big[\frac{m_{a,b}}{m_{a,b} + 1}\Big], 
	\]
	and in particular
	\[
		\frac{1}{2} {\mathbf{P}}_{G,\beta}(m_{a,b} > 0) \leq \langle \sigma^2_{a}\bar{\sigma}^2_{b} \rangle_{G,\beta} \leq  {\mathbf{P}}_{G,\beta}(m_{a,b} >0). 
	\]
\end{lemma}

\begin{figure}
	\begin{subfigure}{.2\textwidth}
		\centering
		\includegraphics[scale =0.11]{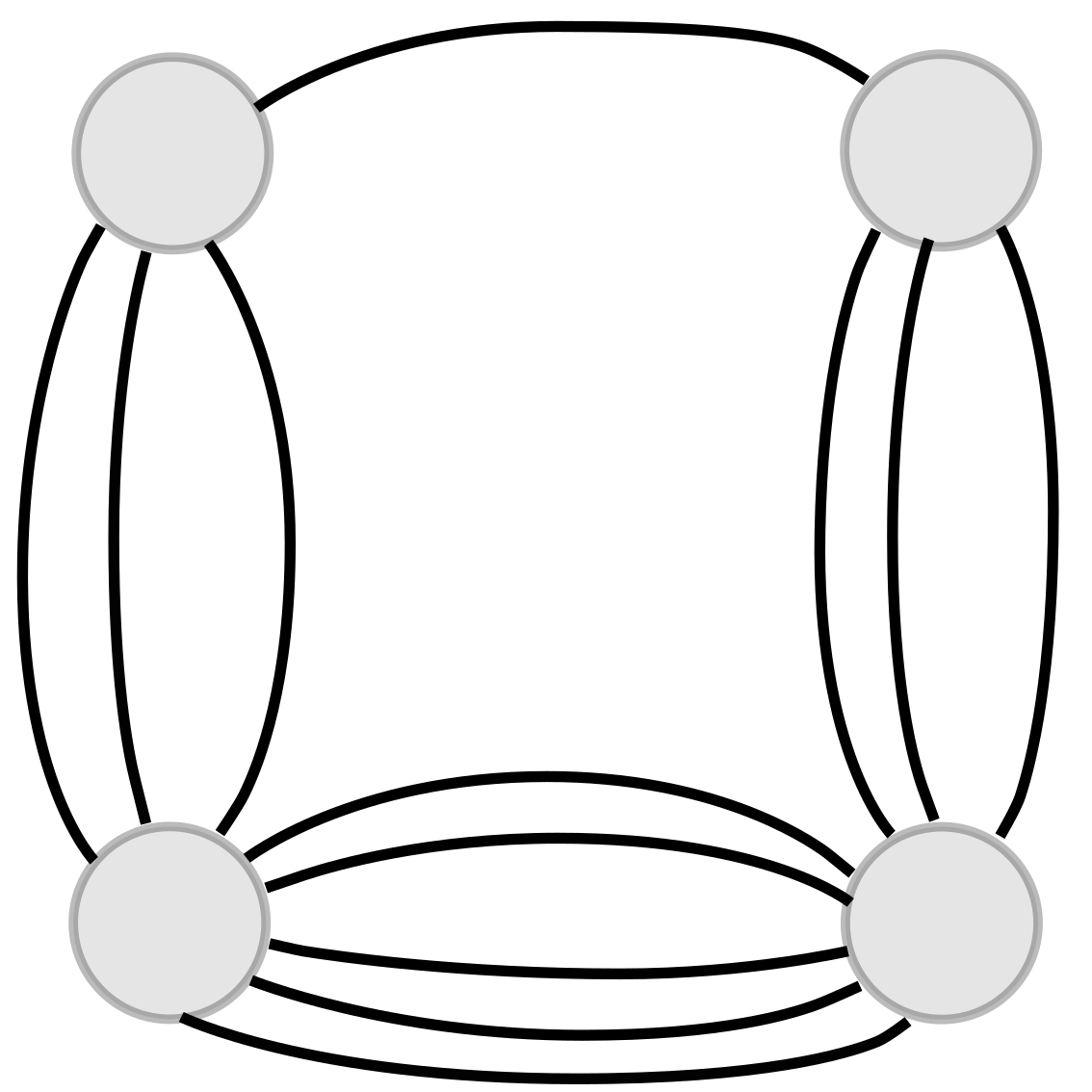}
	\end{subfigure} \hspace{0.5cm}
	\begin{subfigure}{.2\textwidth}
		\centering
		\includegraphics[scale =0.11]{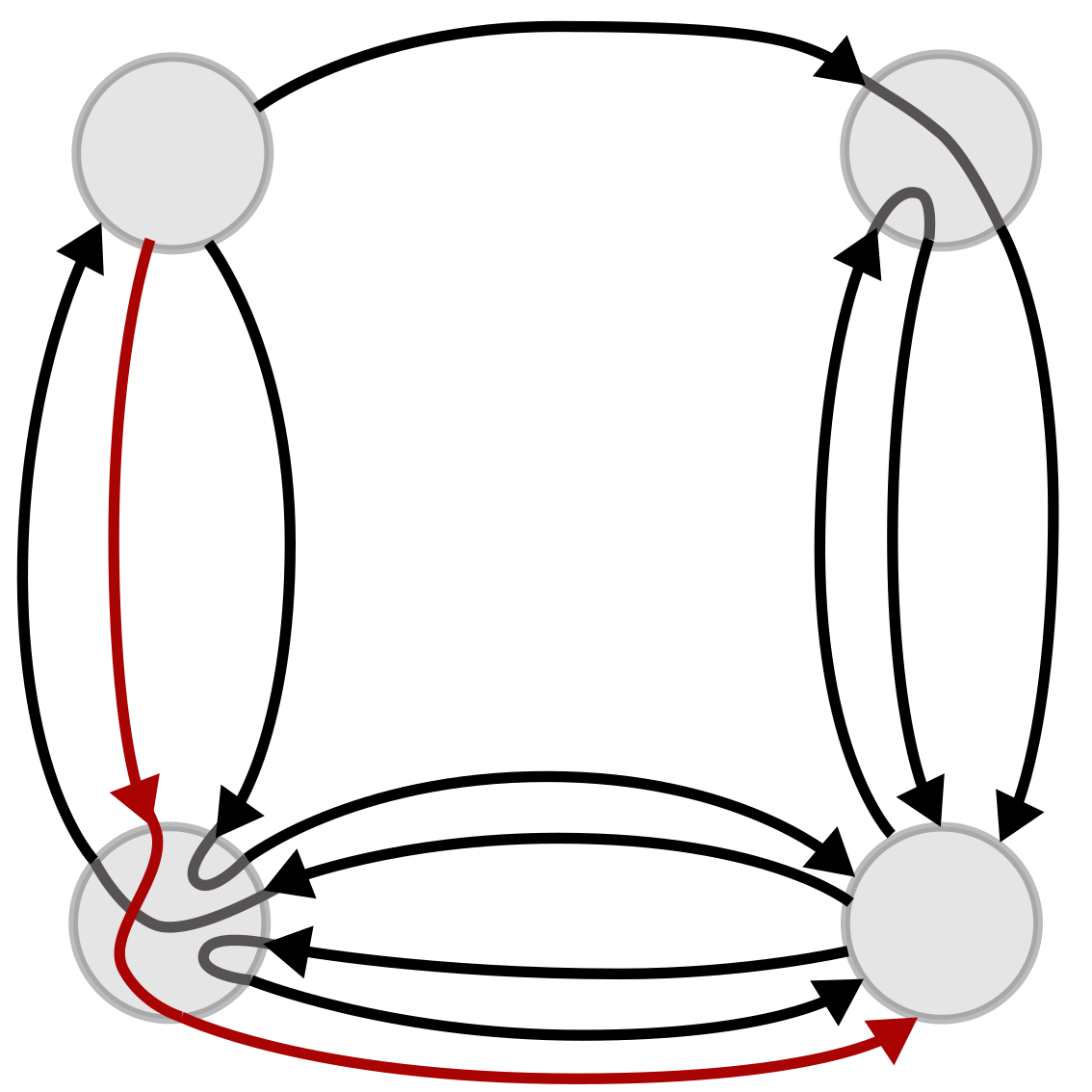} 
	\end{subfigure}  \hspace{0.5cm}
	\begin{subfigure}{.2\textwidth}
		\centering
		\includegraphics[scale =0.11]{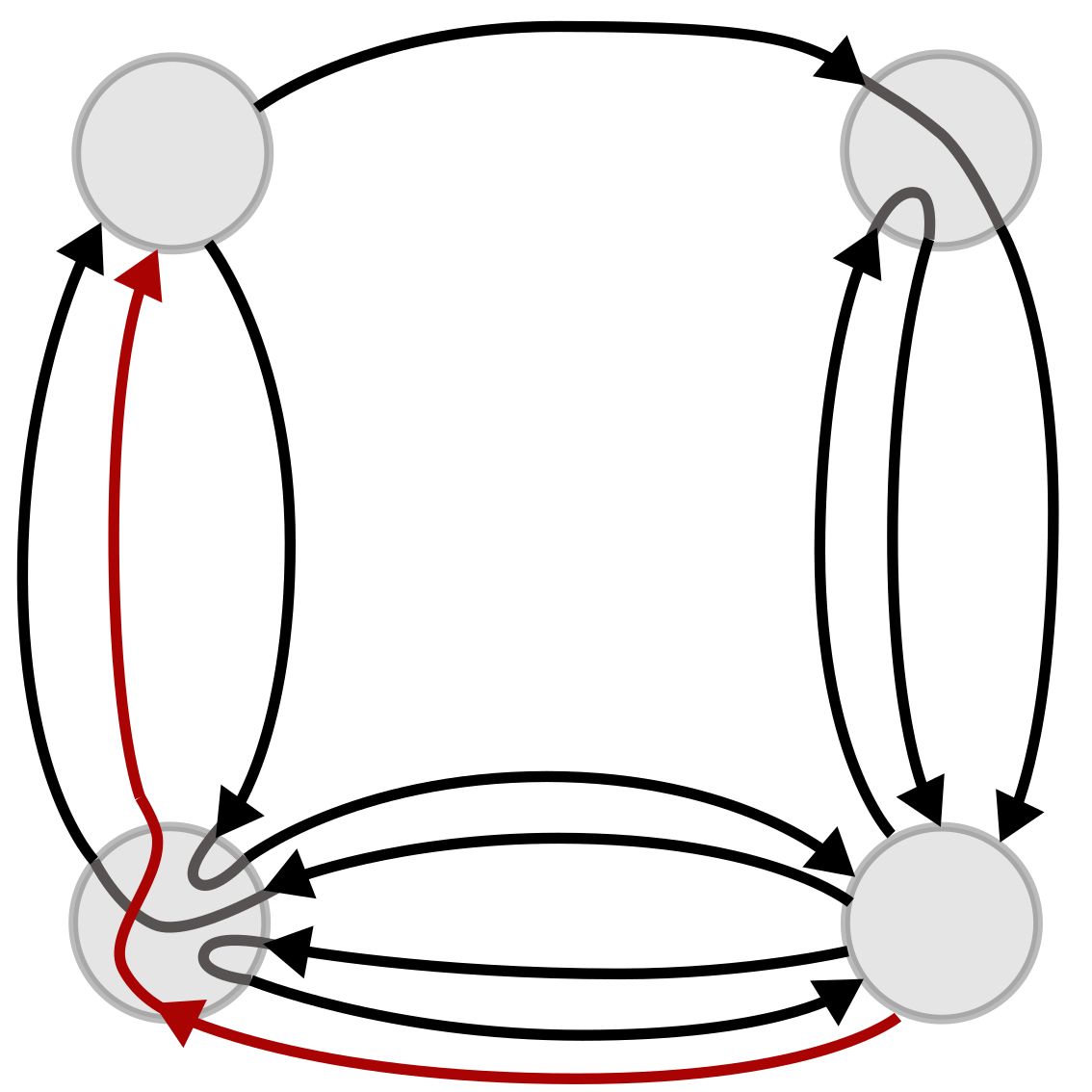} 
	\end{subfigure}  \hspace{0.5cm}
	\begin{subfigure}{.2\textwidth}
		\centering
		\includegraphics[scale =0.11]{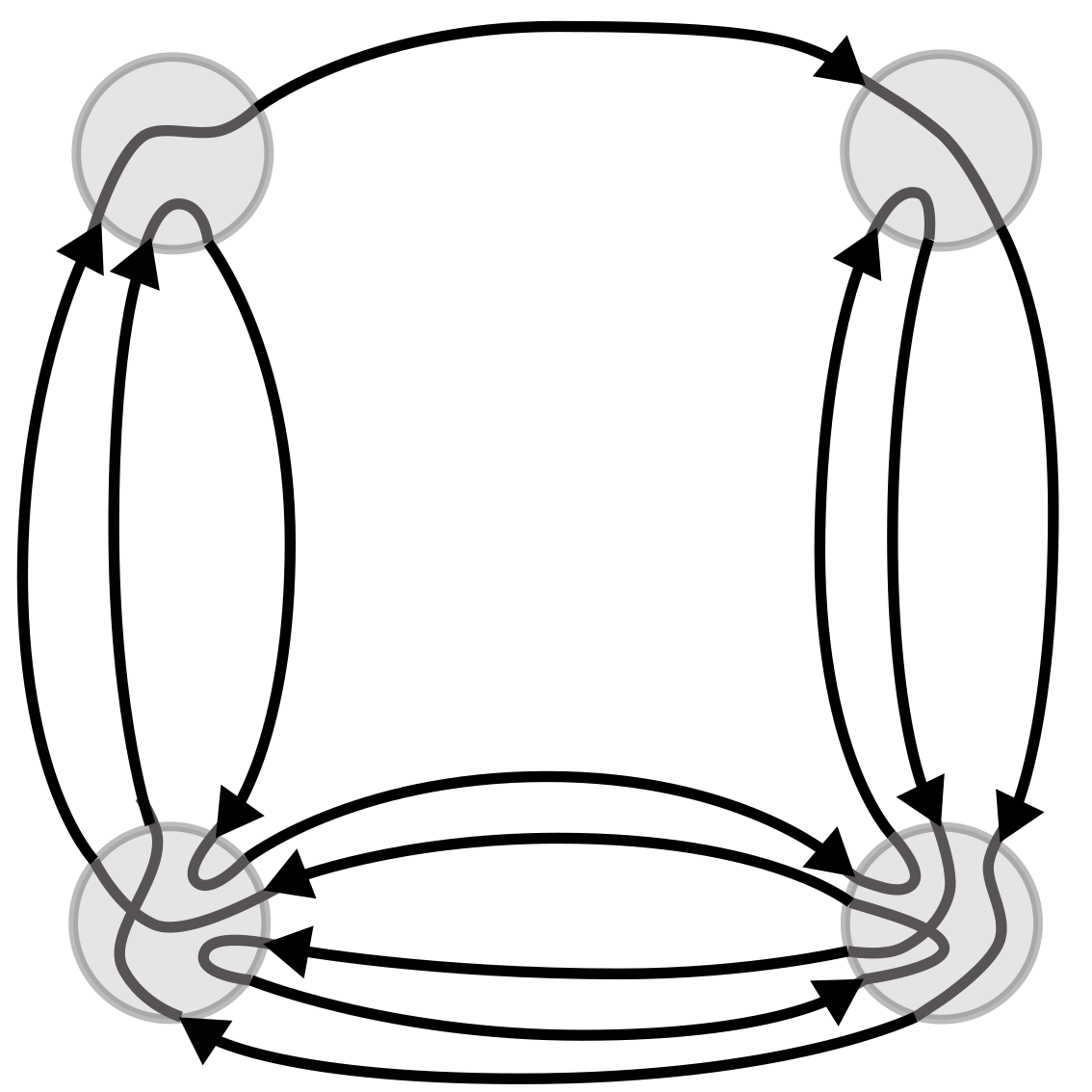}
	\end{subfigure}
	\caption{Left to right: an Eulerian multigraph $\mathcal M$; a loop configuration $\omega\in \mathcal L^{a,b}_{2(\delta_a-\delta_b)}$ on $\mathcal M$ ($a$ is the top left and $b$ the bottom right vertex) together with a path from $a$ to $b$ marked red; a loop configuration $\omega'\in \mathcal L^{a,b}_{0}$ with the path reversed; and one of the final loop configurations $\omega''\in\mathcal L^{\emptyset}_{0}$ corresponding to $\omega'$, i.e., such that $\rho(\omega'')=\omega'$.}
	\label{f:path-switching}
\end{figure}

Let us finish with a number of remarks.
\begin{remark}
We stress again that the crucial property of this loop representation is that the measure $ {\mathbf{P}}_{G,\beta}$ is supported on collections of closed loops, and is independent of the choice of $a$ and $b$.
A similar idea was used by Lees and Taggi~\cite{LeeTag} to study spin $O(n)$ models with an external magnetic field.
Moreover, by Proposition~\ref{prop:netwinding} and Lemma~\ref{L:SingleSwitching}, the random loops under ${\mathbf{P}}_{G,\beta}$ carry probabilistic information about \emph{both} the spin XY model (in terms of correlation functions)
and its dual height function (as an exact coupling). An analogous role for the Ising and Ashkin--Teller model is played by the (double) random current measure that encodes both an integer valued height function and the spin correlations~\cite{DCL,Lis19,LisAT}.
The difference is that for the XY model, the correlations are determined by loop connectivities instead of percolation connectivities. This comparison offers an alternative explanation for the different types of phase transition in discrete and continuous spin systems.
\end{remark}

\begin{remark}\label{rem:BenUel}
The approach above is different from~\cite{Sym,BFS,BenUel,LeeTag} in that in the loop configurations, we never make connections at vertices with sources. 
This leads to different combinatorics than in~\cite{BenUel}, and in particular a more transparent formula for the two-point function. See also Appendix~\ref{sec:doubleswitch} for a different construction
where we allow such connections.  
\end{remark}

\begin{remark}
We call a multigraph $\mathcal M$ \emph{Eulerian} if its degree is even at every vertex.
Another way to sample the loop configuration that easily follows from the above definitions is the following procedure: 
\begin{itemize}
\item First sample an Eulerian submultigraph $\mathcal M$ of $G$ with probability proportional to 
\[
\mathcal E(\mathcal M) \prod_{e\in E} \frac{1} {\mathcal M_e!}\Big(\frac{\beta}{2}\Big)^{\mathcal M_e},
\]
where $\mathcal E(\mathcal M)$ is the number of \emph{Eulerian orientations} of $\mathcal M$, i.e., assignments of orientations to every edge of $\mathcal M$ with an equal number of incoming and outgoing edges at every vertex.
\item Then choose uniformly at random an Eulerian orientation of $\mathcal M$.
\item Finally, at each vertex, independently of other vertices, connect the incoming edges with the outgoing edges uniformly at random.
\end{itemize}
\end{remark}

\begin{remark}
Using the same argument as above one obtains the following formula for higher power two-point functions. For $k\geq1$, we have
	\[
		\langle \sigma^{2k}_{a}\bar{\sigma}^{2k}_{b} \rangle_{G,\beta} = {\mathbf{E}}_{G,\beta}\Big[\frac{(m_{a,b})_k}{(m_{a,b} + k)_k}\Big],
	\]
where $(m)_k = m(m - 1)\cdots (m - k+ 1)$ is the falling factorial. One can also consider multi-point functions and get more complicated loop representation formulas. 
\end{remark}

\begin{remark}
This representation is valid on any, not necessarily planar, graph, and it is known that the XY model exhibits long-range order in dimension greater than two~\cite{FSS}. The disorder-order transition should coincide with the onset of infinite loops (biinfinite paths) on the current. 
The alternative heuristic for the lack of symmetry breaking in two dimensions arising from this picture is that planar simple random walk is recurrent (an hence does not produce infinite loops).
\end{remark}

\begin{remark}
The isomorphism theorem of Le Jan~\cite{LeJan} says that the discrete complex Gaussian free field can be coupled with a Poissonian collection of random walk loops, the so called \emph{random walk loop soup}, 
in such a way that one half of the square of the absolute value of the field is equal to the total occupation time of the random walk loops. On the other hand, it is immediate that conditioned on the absolute value of the field,
its complex phase is distributed like the XY model with coupling constants depending on this absolute value. 
With some work, e.g.\ using~\cite{LeJan1}, one can show that under this conditioning the random walk loops have the same distribution as the loops described above.
\end{remark}

\section{Delocalisation implies no exponential decay} \label{sec:noexpdec}
In this section we prove that if the height function delocalises, then the spin correlations are not summable along certain sets of vertices. In the next section, we will show how to apply this 
together with the delocalisation results of Lammers~\cite{Lammers} to deduce a BKT-type phase transition in a wide range of periodic planar graphs.

Suppose $\Gamma=(V,E)$ is a translation invariant planar graph, and write
\begin{align} \label{eq:definfvol}
\langle \sigma_a\bar\sigma_{b} \rangle_{\Gamma,\beta} = \lim_{G\nearrow \Gamma}\langle \sigma_a\bar\sigma_{b} \rangle_{G,\beta}
\end{align}
for the infinite volume two-point function,
where the limit is taken along any increasing sequence of subgraphs $G$ exhausting $\Gamma$. That this is well defined is guaranteed by the fact that the sequence is nondecreasing, i.e., 
$\langle \sigma_a\bar\sigma_{b} \rangle_{G,\beta}\leq \langle \sigma_a\bar\sigma_{b} \rangle_{G',\beta}$ if $G$ is a subgraph of $G'$, which in turn is a classical consequence of the Ginibre inequality~\cite{Gin}.
\begin{definition}Let $0$ be a distinguished face of $\Gamma$. A bi-infinite self-avoiding path in $\Gamma$ that goes through at least one edge 
	incident to $0$ is called a \emph{cut (at $0$)}. Note that a cut $L$ naturally splits into two infinite sets of vertices $L_{ +}$ and $L_{ -}$ with the property that any cycle in $\Gamma$ that surrounds $0$ must intersect both $L_+$ and $L_-$.
\end{definition}
The main quantity of interest for us will be the sum of correlations along cuts. To be more precise for $\varepsilon>0$,
let
\begin{align} \label{E: chi}
		\chi^{\epsilon}_{\Gamma,\beta}(L)=\sum_{a \in L_{ +}, b \in L_{ -}} (\langle \sigma_{a} \overline{\sigma}_{b} \rangle_{\Gamma,\beta})^{2-\varepsilon}. 
\end{align}

\begin{proposition} \label{T:BKT-transition}
	For every $\epsilon > 0$, there exists $C=C(\epsilon,\beta,\Gamma) < \infty$ such that for all finite subgraphs $G$ of $\Gamma$ containing $0$, we have
		\[
		\E_{ G, \beta}[ |h(0)|] \leq C\inf_{L}\chi^{\epsilon}_{\Gamma,\beta}(L),
	\]
	where the infimum is over all cuts at $0$.
\end{proposition}

Before presenting the proof, let us mention that a direct corollary of this proposition is the following.
A natural example of a cut is any path that stays at a constant distance from a straight line going through $0$.
In this case it is easy to see that $\chi^{\epsilon}_{\Gamma,\beta}(L)$ is finite whenever there is exponential decay of spin correlations. We can now state the main conclusion of this section. 
\begin{corollary}\label{cor:implication}
If the height function delocalises in the sense of Theorem~\ref{T: dichotomy}, then \[
\chi^{\epsilon}_{\Gamma,\beta}(L)=\infty
\] 
for all $\varepsilon>0$ and all cuts $L$ at $0$. In particular the two-point function does not decay exponentially fast with the distance between the vertices. 
\end{corollary}
\begin{proof}
We know that situation $(i)$ from Theorem~\ref{T: dichotomy} does not happen. This means that $\sup_n\E_{B_n,\beta}[ |h(0)|] =\infty$, 
and the claim follows directly from Proposition~\ref{T:BKT-transition}.
\end{proof}

\begin{remark}
One naturally expects that the localisation-delocalisation phase transition for the height function happens at the same temperature as the BKT transition for the XY model.
The remaining part of this prediction is therefore to show that if the spin correlations do not decay exponentially, then the height function delocalises. We do not do this in this article.
\end{remark}

Recall that $m_{a,b}$ is the number of paths (pieces of loops) in a loop configuration that go from $a$ to $b$.
We will need the following lemma.
\begin{lemma} \label{L: M_xy L2-bounded}
	For all $\beta > 0$ and $p> 1$, there exists a $C_p < \infty$ such that for all finite graphs $G=(V,E)$ and all $a,b\in V$,
	\[
		\mathbf{E}_{G,\beta}[m_{a,b}] \leq C_p \deg_{G}(a)\big(\mathbf{P}_{G,\beta}(m_{a,b} > 0) \big)^{\frac{1}{p}}.
	\]
\end{lemma}
\begin{proof}
	Fix $\beta > 0$, $G=(V,E)$ and $a,b\in V$, and let $\omega \in \c{L}_0$ be a loop configuration on $G$. 
	Denote by $\omega_e$, the number of visits of all loops in $\omega$ to an undirected edge $e\in E$. If there are $m \geq 1$ paths going from $a$ to $b$ in $\omega$, then in particular $\sum_{c\sim a} \omega_{\{ a, c\}} \geq m$. 
	This implies that
		\begin{align*}
			\mathbf{E}_{G,\beta}[m_{a, b}] \leq \mathbf{E}_{G,\beta}\Big[\sum_{c \sim a}\omega_{\{ a, c\}} \id\{m_{a, b }> 0\} \Big] \leq \deg_G(a) \max_{c \sim a}\mathbf{E}_{G,\beta}[\omega_{\{ a, c\}} \id\{m_{a, b} > 0\}].
		\end{align*}
	Applying H\"{o}lder's inequality gives
		\begin{align*}
			\mathbf{E}_{G,\beta} [\omega_{\{a, c\}}  \id\{m_{a, b} > 0\}] \leq \big(\mathbf{E}_{G,\beta}[\omega_{\{ a, c\}}^q]\big)^{1/q}\mathbf{P}_{G,\beta}(m_{a, b} > 0)^{1/{p}},
		\end{align*}
	where $1/p+1/q=1$. We now notice that by definition, $\omega_e$ under $\mathbf{P}_{G,\beta}$ has the same distribution as the amplitude $|\n|_e$ under $\mathbb{P}_{G,\beta}$. Therefore, to finish the proof it is enough to show that for all $p>1$, there exists $C_p<\infty$ depending on $\beta$ but independent of $G$ such that
	\begin{align} \label{eq:amplitudemoment}
	\mathbb{E}_{G,\beta}[|\n|^p_e]\leq C_p.
	\end{align}
	We postpone the proof of this bound to Lemma~\ref{L: integr of gradient} and Lemma~\ref{L: bd moments Xk}.
\end{proof}
The last ingredient that we will need is the following inequality
\begin{lemma} \label{lem:squares}
For any $a,b\in V$, we have
\begin{align*} \label{eq:2to1}
\langle \sigma^2_{a}\bar{\sigma}^2_{b} \rangle_{G,\beta} \leq2 \langle \sigma_{a}\bar{\sigma}_{b} \rangle^2_{G,\beta}.
\end{align*}
\end{lemma}
\begin{proof}
A version of the Ginibre inequality (see e.g.~\cite{BLU}) says that
\begin{align*}
\big\langle \Im(\sigma_a) \Im(\sigma_b)\Re(\sigma_a)\Re(\sigma_b)\big\rangle_{G,\beta} &\leq \big\langle  \Im(\sigma_a) \Im(\sigma_b)\big\rangle_{G,\beta}  \big\langle  \Re(\sigma_a)\Re(\sigma_b)\big\rangle_{G,\beta}.
\end{align*}
which after rearrangement gives the desired inequality.
\end{proof}
We note that the constant in the inequality above can be improved to $1$ using our switching techniques from Appendix~\ref{sec:doubleswitch} (see Remark~\ref{rem:squares}).

We are now ready to prove the main theorem.

\begin{proof}[Proof of Proposition~\ref{T:BKT-transition}]
Fix a finite subgraph $G$ and a cut $L$. By Proposition~\ref{prop:netwinding} the height function $h(0)$ under $\mathbb{P}_{G,\beta}$ has the sam law as $W(0)$ -- the total net winding {around $0$} of all loops in a loop configuration -- drawn according to $\mathbf{P}_{G,\beta}$.
	Moreover, any piece of a loop that adds to the winding (in any orientation) must intersect both $ L_{+}$ and $ L_{ -}$ by definition of a cut. Therefore, taking $p=2/(2-\varepsilon)$, we have
	\begin{align*}
			\E_{G,\beta}[|h(0)|] &= \mathbf E_{G,\beta}[|W(0)|]  \\
		&\leq \sum_{a\in L_{+}, b \in L_{ -}} \mathbf E_{G,\beta}[m_{a,b}] \\ 
		&\leq \tilde C\sum_{a \in L_{+}, b \in L_{ -}} ( \mathbf P_{G,\beta}(m_{a,b} > 0))^{1/p} \\
		& \leq 2\tilde C\sum_{a \in L_{+}, b \in L_{ -}} (\langle \sigma^2_{a}\bar{\sigma}^2_{b} \rangle_{G,\beta})^{1-\varepsilon/2} \\
		& \leq 4\tilde C\sum_{a \in L_{+}, b \in L_{ -}} (\langle \sigma_{a}\bar{\sigma}_{b} \rangle_{G,\beta})^{2-\varepsilon} \\
		& \leq C\chi^{\epsilon}_{\Gamma,\beta}(L).
	\end{align*}
	where the third line follows from Lemma \ref{L: M_xy L2-bounded}, the forth one from Lemma~\ref{L:SingleSwitching}, the fifth one from Lemma~\ref{lem:squares}, and the last one from~\eqref{eq:definfvol}.
	This completes the proof. 
\end{proof}

It therefore remains to show \eqref{eq:amplitudemoment}, which will directly follow from Lemma~\ref{L: integr of gradient} and Lemma~\ref{L: bd moments Xk} below. To that end, define for $k \in \NN$ and $\beta > 0$, a random variable $Y_{k}$ by 
\[
	\PP_{\beta}(Y_{k} = i) \propto\frac{1}{i!(i + k)!} \big(\tfrac{\beta}{2}\big)^{2i + k},
\]
so that the normalizing constant is $I_k(\beta)$. For $e=vv'$, let 
\[
|\nabla h|_e= |\n_{(v, v')} - \n_{(v', v)}|
\]
be the absolute value of the gradient of the height function across the dual edge $e^\dagger$. Note that the random variables $(X_{e} = X_e(\n))_{e \in E}$ defined through
\[
	X_{e} = \frac{|\n|_e - |\nabla h|_e}{2}
\]
have the same distribution as $Y_{|\nabla h|_e}$. Moreover, conditionally on $|\nabla h|$, they are an independent family. To show \eqref{eq:amplitudemoment} it is enough to bound the moments of $|\nabla h|_e$ and $X_e$ separately, which we will now do.

\begin{lemma} \label{L: integr of gradient}
	For all $\beta > 0$ and all $r \in \mathbb N$, there exists a $C_r < \infty$ such that for all finite planar graphs $G = (V, E)$ and all $e \in E$, 
	\[
		\E_{G,\beta}[|\nabla h|_e^r] \leq C_r.
	\]
\end{lemma}
\begin{proof}
	Fix a finite planar graph $G$, and let $ e=vv' \in E$. 
	We write $G \setminus e$ for the graph without the edge $e$. For $l\in \mathbb Z$, we define $\Omega_{ l}(G)  =\{ \n \text{ on } G: \delta \n =l(\delta_v-\delta_{v'})\}$, and
	\[
		Z^l_G = \sum_{\n \in \Omega_{r} (G)} w_{\beta}(\n), 
	\]
	 and analogously $Z^l_{G \setminus e}$. Similarly to \eqref{eq:currentexp}, we get from the current expansion of correlation functions that
	\[
		\langle \sigma_v^l \bar\sigma_{v'}^{l} \rangle_{G\setminus e, \beta} = \frac{Z_{G \setminus e}^{l}}{Z^0_{G \setminus e}}.  
	\]
By the definition of the height function and currents, we therefore have 
\begin{align*}
 \PP_{G,\beta}(|\nabla h|_e= l) = I_l(\beta)\frac{(Z^l_{G\setminus e} +Z^{-l}_{G\setminus e}) }{Z^0_G}=2I_l(\beta)\frac{Z^l_{G\setminus e} }{Z^0_{G\setminus e} } \frac{Z^0_{G\setminus e} }{Z^0_{G} } \leq 2 I_l(\beta) \leq 2I_0(\beta) \frac{ \beta^l}{2^ll!},
\end{align*}
where we used the obvious bounds $\langle \sigma_v^l \bar\sigma_{v'}^{l} \rangle_{G\setminus e, \beta} \leq 1$, and ${Z^0_{G\setminus e} }/{Z^0_{G}} \leq 1$, and the last inequality follows easily from the definition of $I_r(\beta)$.
Finally, 
	\begin{align*}
		\E_{G,\beta}[|\nabla h|_e^r] &= \sum_{l \geq 1} l^r\PP_{G,\beta}(|\nabla h|_e = l) \leq 2 I_0(\beta)\sum_{l \geq 1}l^r\frac{\beta^l}{2^ll!}=: C_r < \infty.
	\end{align*}
	The last bound is independent of $G$ and $e$ which completes the proof.
\end{proof}

\begin{lemma} \label{L: bd moments Xk}
	For all $\beta > 0$ and all $r \in \mathbb N$, there exists a $\tilde C_r < \infty$ such that for all finite planar graphs $G = (V, E)$ and $e \in E$, 
	\[
		\E_{G,\beta}[|X_e|^r] \leq \tilde C_r. 
	\]
\end{lemma}
\begin{proof}
For two nonnegative integers $i, r$, let $(i)_r = i(i - 1)\cdots (i - r + 1)$ be the falling factorial with the convention that $(i)_0=1$. Note that $(i)_r=0$ whenever $ i < r$. It will be convenient to look at the falling factorial moments.
First note that  by definition of $Y_k$,  
	\begin{align*}
		\E_\beta[(Y_k)_r] &= \frac1{I_k(\beta)}\sum_{i \geq 0} \frac{(i)_r}{i!(i + k)!} \big( \tfrac{\beta}{2}\big)^{2i + k} =   \frac{\big(\tfrac{\beta}{2}\big)^{r}}{I_k(\beta)}\sum_{i \geq 0} \frac1{i!(i + k + r)!}\big( \tfrac{\beta}{2}\big)^{2i + k + r} = (\tfrac{\beta}{2}\big)^{r}  \frac{I_{k+r}(\beta)}{I_k(\beta)}.
	\end{align*}
	By the Tur\'{a}n inequality \eqref{E: Turan}, the map $k \mapsto I_{k+1}(\beta) / I_k(\beta)$ is decreasing and hence
	\[
		\E_\beta[(Y_k)_{r}] = \big(\tfrac{\beta}{2}\big)^{r } \frac{I_{k + r }(\beta)}{I_k(\beta)} \leq \big(\tfrac{\beta}{2}\big)^{r }\frac{I_{r}(\beta)}{I_0(\beta)} =: C.
	\]
	Now note that $(i)_r \geq |i-r|^r$ when $i\geq r$, and hence $i^r\leq2^{r-1}( |i-r|^r +r^r)\leq 2^{r}( (i)_r +r^r)$. Finally
\[
		\E_{\beta}[|X_e|^r \mid |\nabla h|_e = k] = \E_{ \beta}[|Y_k|^r] \leq 2^{r}( C +r^r) := \tilde C_r,  
	\]
	where the last bound does not depend on $k$. Integrating over the possible values of $|\nabla h|_e$ concludes the proof. 
\end{proof}

\section{Existence of phase transition in the XY model} \label{sec:existence}
In this section, we prove that for all translation invariant planar graphs $\Gamma = (V, E)$, the XY model undergoes a non-trivial phase transition in terms of the quantity $\chi^{\varepsilon}_{\beta}(L)$.  
As before, let $0$ denote an arbitrary distinguished face of $\Gamma$. We define
\begin{align*}
\beta_0=\inf\{\beta>0: \text{for all ${\varepsilon>0}$ and all cuts $L$ at $0$, $\chi^{\varepsilon}_{\beta}(L) =\infty$}\}.
\end{align*}
\begin{theorem} \label{T:BKT planar lattice}
	Let $\Gamma$ be as above. Then $\beta_0<\infty$.
\end{theorem}
By Corollary~\ref{cor:implication} it is enough to show that for any such $\Gamma$, there exists a finite $\beta_0 > 0$ such that the associated height function delocalises in the sense that there are no translation invariant Gibbs measures on the dual $\Gamma^\dagger$. 
We first implement this strategy for triangulations, where delocalisation can be shown directly using the general result of Lammers~\cite{Lammers} (Theorem~\ref{thm:lammers}).
\begin{proof}[Proof of Theorem \ref{T:BKT planar lattice} for triangulations]
	Let $\Gamma$ be a translation invariant triangulation. Note that condition \eqref{E: exited potential} in our case is equivalent to ${I_{1}(\beta)}/{I_0(\beta)} \geq \frac{1}{2}$. 
	It is known that this fraction converges to $1$ as $\beta \to \infty$ (see for example \cite{Segura}), and therefore in light of Theorem~\ref{thm:lammers}, there are no translation invariant Gibbs measures for $\beta$ large enough. 
\end{proof}

To extend beyond triangulations, we will use a different approach. We stress that in particular, we will not show delocalisation of the height function on graphs that are not triangulations.
Instead, we exploit monotonicity in coupling constants to bound from below the spin correlations on an arbitrary translation invariant graph by correlations on a modified graph that is a triangulation. We explain this procedure in detail for the square lattice, and briefly mention the extension to other lattices at the end. 

In what follows, we will need the following well known monotonicity of spin correlations that is a classical consequence of the Ginibre inequality~\cite{Gin}.
\begin{lemma} \label{L: xy increasing cor}
	For each (infinite or finite) graph $G = (V, E)$, $\beta > 0$, and $e\in E$, the function
	\[
		J_e \mapsto \langle \sigma_v\bar \sigma_{v'} \rangle_{G,\beta} 
	\]
	is nondecreasing. 
\end{lemma}

\begin{figure}
	\begin{subfigure}{.2\textwidth}
		\centering
		\includegraphics[scale =0.11]{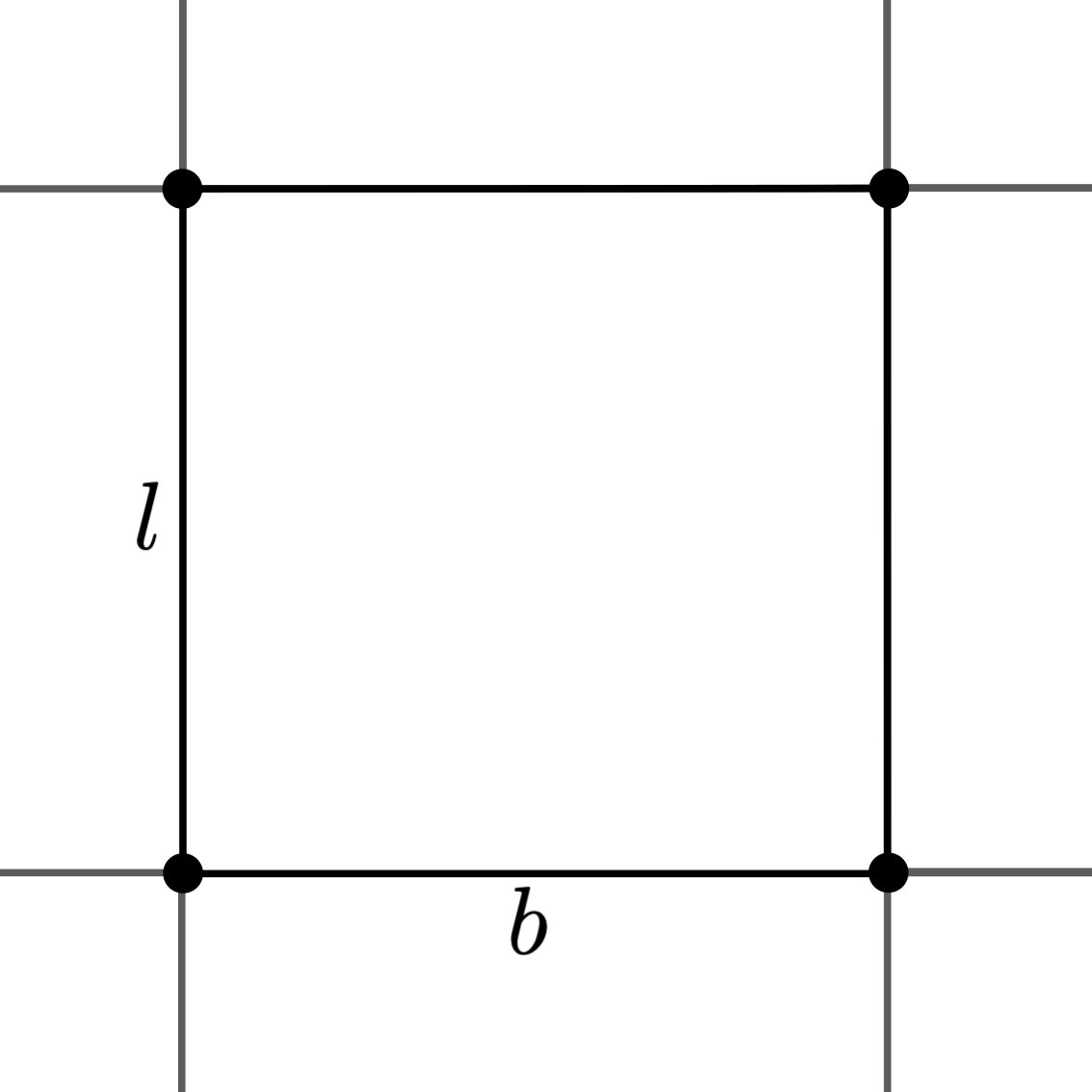}
	\end{subfigure} \hspace{0.5cm}
	\begin{subfigure}{.2\textwidth}
		\centering
		\includegraphics[scale =0.11]{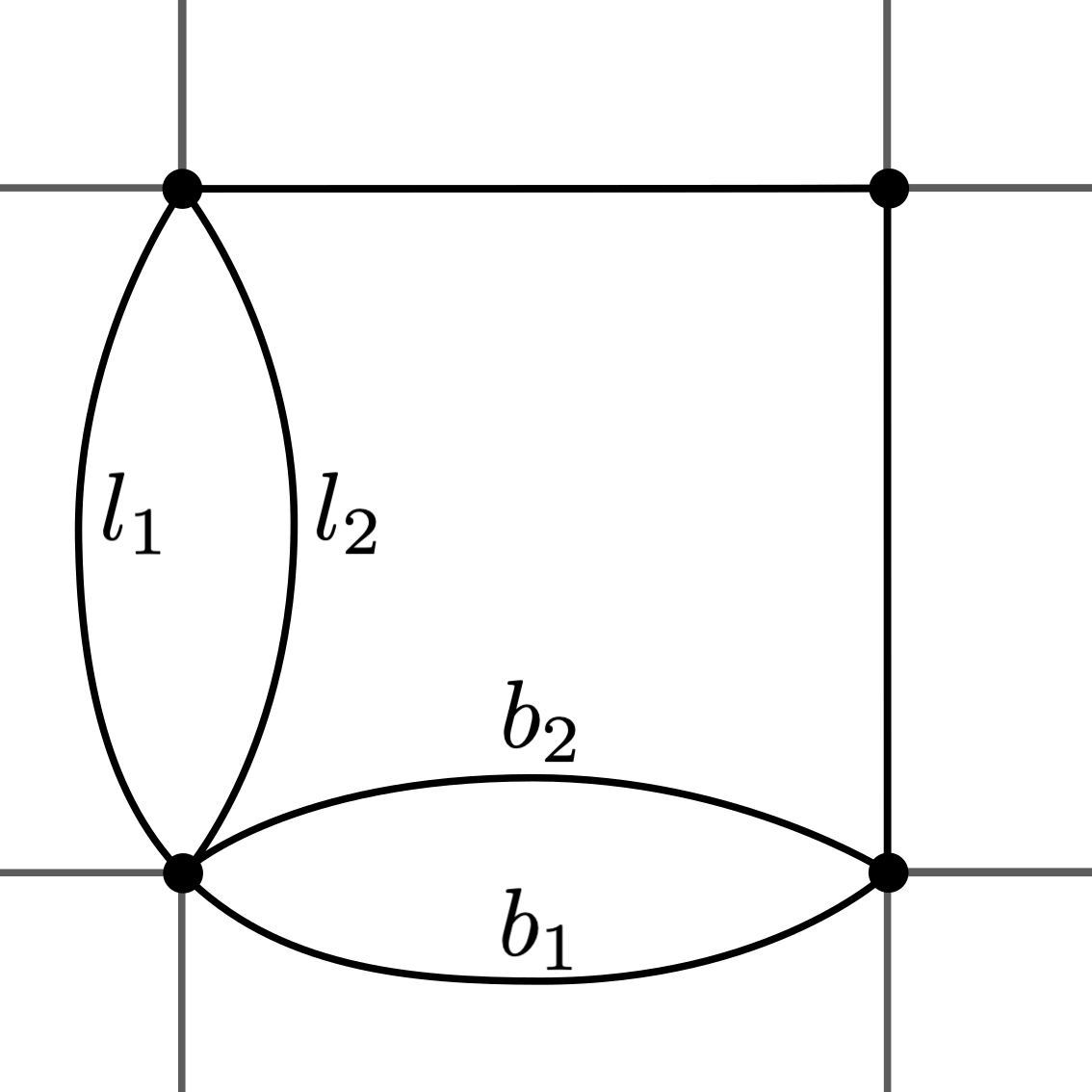} 
	\end{subfigure}  \hspace{0.5cm}
	\begin{subfigure}{.2\textwidth}
		\centering
		\includegraphics[scale =0.11]{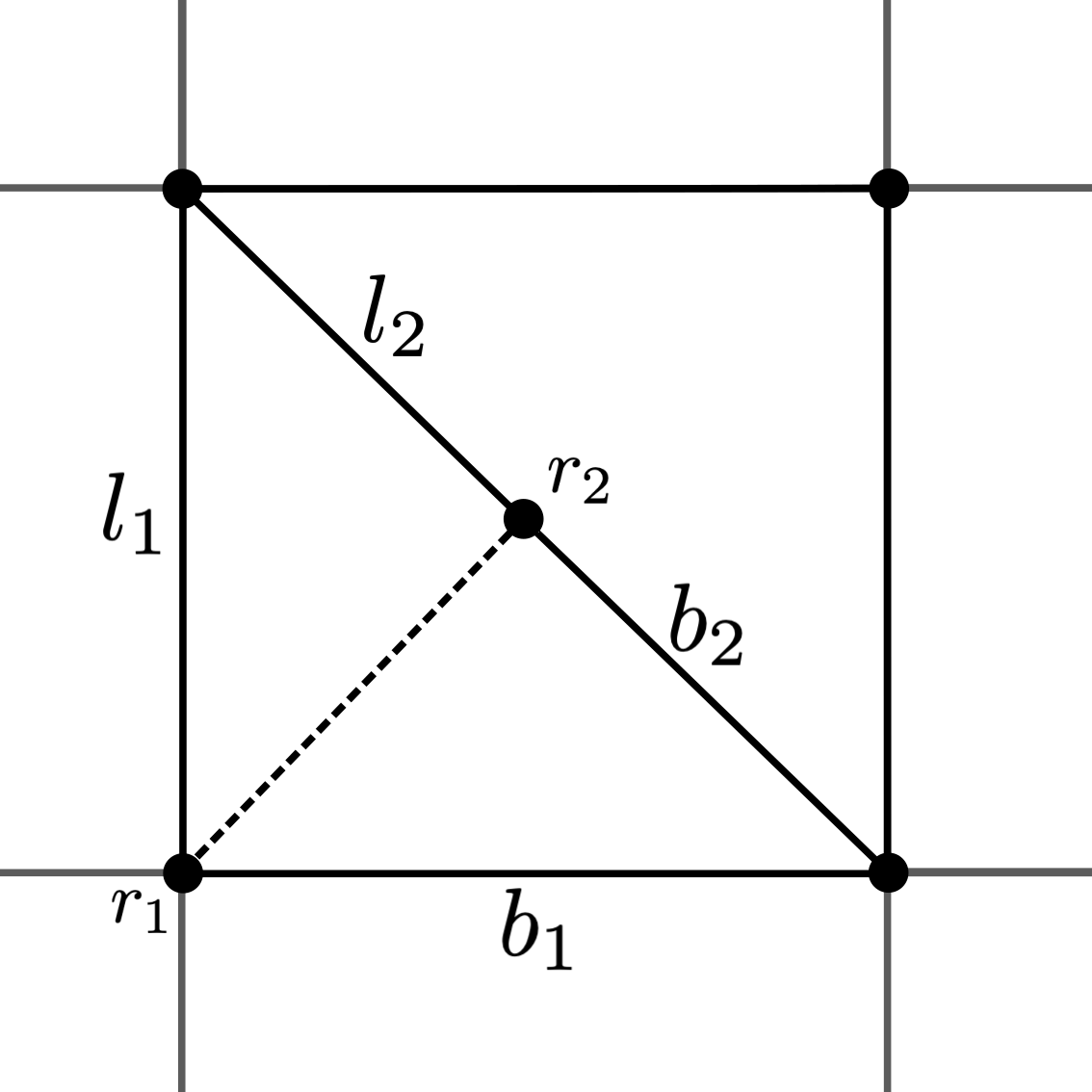} 
	\end{subfigure}  \hspace{0.5cm}
	\begin{subfigure}{.2\textwidth}
		\centering
		\includegraphics[scale =0.11]{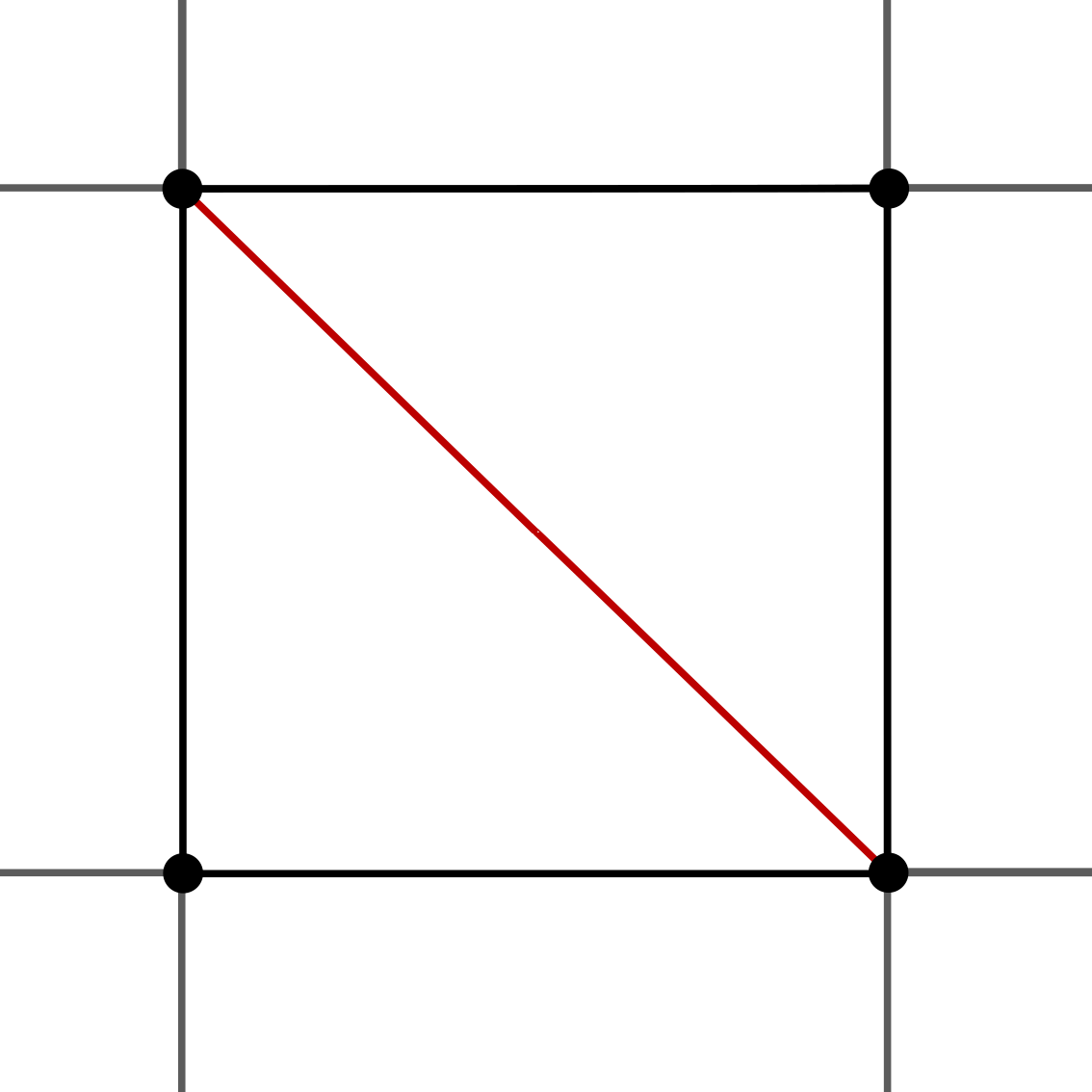}
	\end{subfigure}
	\caption{The transformation to a triangulation. The red edge on the right is the edge with different potential.}
	\label{f: square lattice d3}
\end{figure}

\begin{proof}[Proof of Theorem \ref{T:BKT planar lattice} for the square lattice]
	Let $\Gamma = (V, E)$ denote the square lattice. 
	
	In order to use \eqref{E: exited potential}, we need to transform $\Gamma$ into a triangulation. See Figure \ref{f: square lattice d3} for guidance. Fix a square and double the bottom and left edge and put coupling constants $\beta / 2$ on the doubled edges instead of $\beta$. Next, double the common vertex of the left and bottom edge and add an additional edge $e$, on which we set the coupling constant to infinity. This does not change the distribution of the spins.
	Finally, set the coupling constant on the edge $e$ to $0$, which is equivalent to removing the edge from the square, and repeat the procedure for all other squares. In this way, we obtain a new lattice $\Gamma'$, which consists of squares with a diagonal on which there is an additional vertex. Note that all coupling constants are now equal to $\beta / 2$. By Lemma \ref{L: xy increasing cor}, 
	\begin{align} \label{eq:cordec}
		\langle \sigma_a \bar{\sigma}_{b} \rangle_{\Gamma, \beta} \geq \langle \sigma_a \bar{\sigma}_{b} \rangle_{\Gamma', \beta/2}
	\end{align}
	for all pairs of vertices $a,b$ in $\Gamma$, using the natural embedding of $\Gamma$ on $\Gamma'$. 
	
	Since $\Gamma'$ is a translation invariant graph, the dichotomy statement of Theorem \ref{T: dichotomy} holds. 	
	To show that there are no translation invariant Gibbs measures for the associated height function, notice that the dual $(\Gamma')^\dagger$ of $\Gamma'$ (after collapsing the doubled edges to a single edge) is trivalent. 
	Moreover, the height function on any finite subgraph of $(\Gamma')^\dagger$ has a potential given by $\c{V}_e' = \c{V}^{\beta/2}_e$ for the nondiagonal edges and $\c{V}'_e = 2\c{V}^{\beta/2}_e$ otherwise, and the potential $\c{V}'$ satisfies Lammers' condition \eqref{E: exited potential} precisely when $\left({I_1(\beta/2)}/{I_0(\beta / 2)}\right)^2 \geq \frac{1}{2}$.
	Since the fraction on the left-hand side tends to $1$ as $\beta \to \infty$, we can choose $\beta$ large enough so that there are no translation invariant Gibbs measures for the height function on $(\Gamma')^\dagger$. 
	
	Note that every cut on $\Gamma$ embeds naturally as a cut on $\Gamma'$. Therefore, by Proposition~\ref{T:BKT-transition} together with \eqref{eq:cordec}, we have that
	for each cut $L$ on $\Gamma$ and each $\epsilon > 0$, 
	\[
	\chi_{\Gamma,\beta}^\epsilon(L)\geq \chi_{\Gamma',\beta/2}^\epsilon(L) = \infty.
	\]
	 This finishes the proof.
\end{proof}

To extend this proof to general graphs, we make each face into a triangulation by ``zig-zagging'' (see Figure \ref{f: face degree to 3 general}).
\begin{figure}
	\begin{subfigure}{.2\textwidth}
		\centering
		\includegraphics[scale =0.1]{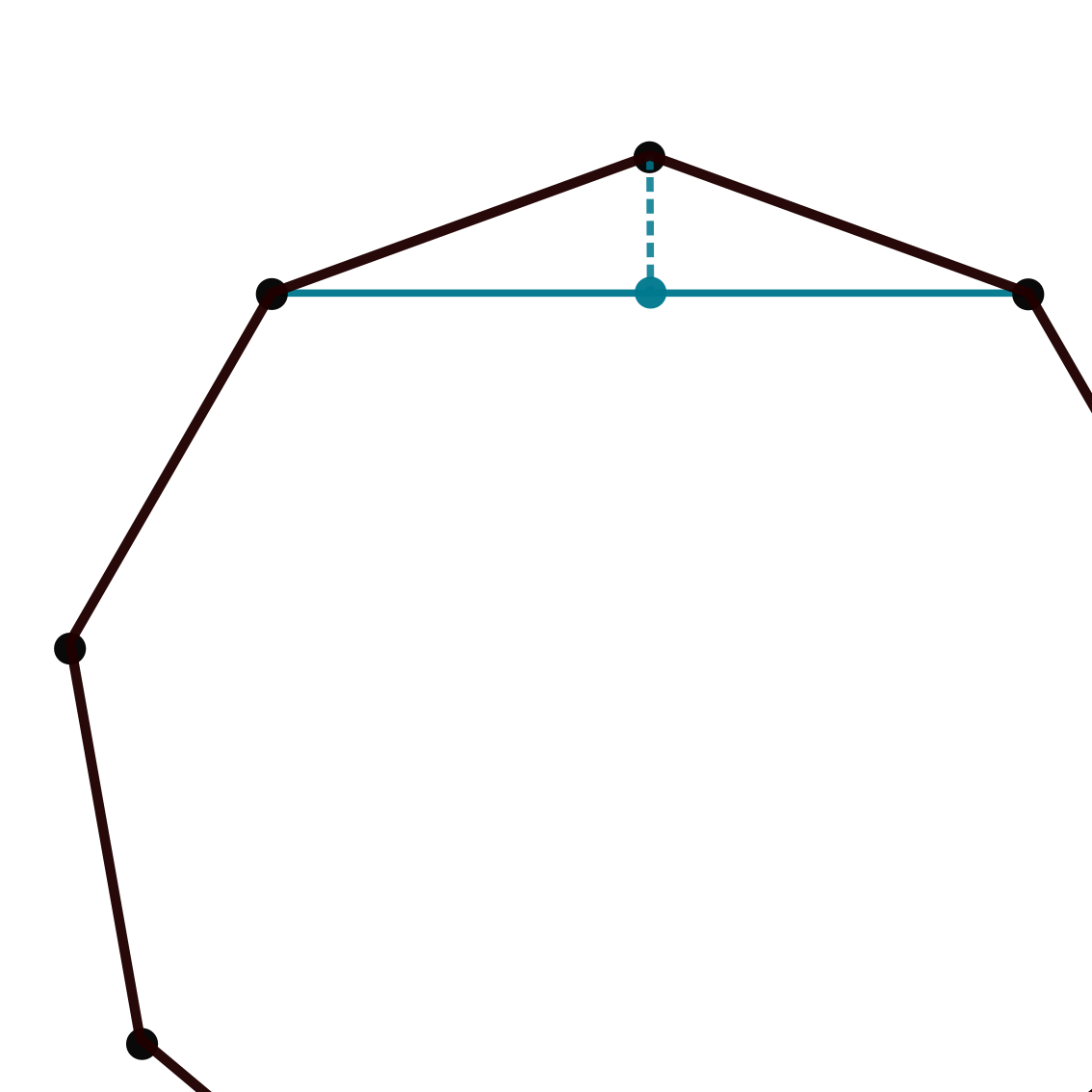}
	\end{subfigure}
	\begin{subfigure}{.2\textwidth}
		\centering
		\includegraphics[scale =0.1]{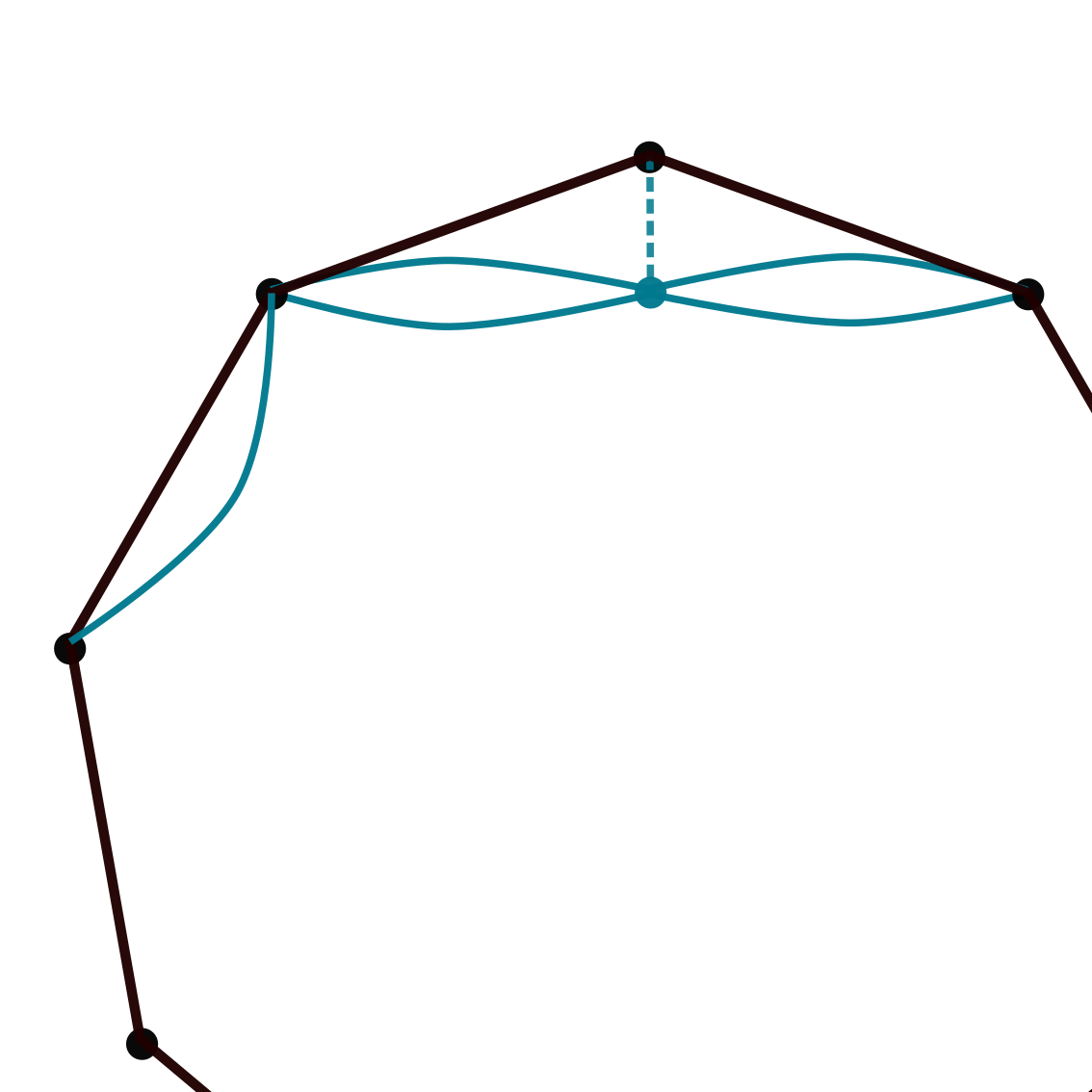}
	\end{subfigure}
	\begin{subfigure}{.2\textwidth}
		\centering
		\includegraphics[scale =0.1]{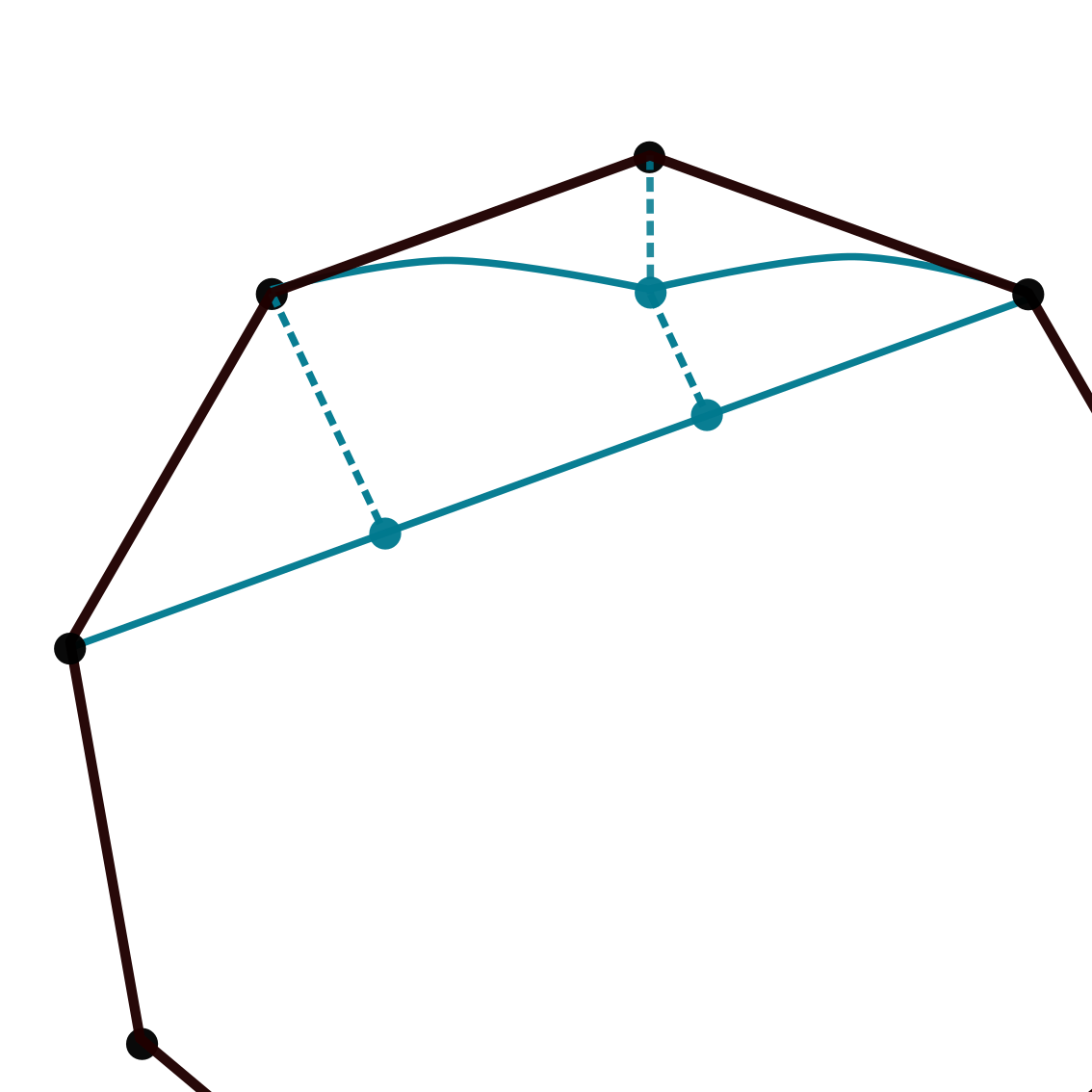}
	\end{subfigure}
	\begin{subfigure}{.2\textwidth}
		\centering
		\includegraphics[scale =0.12]{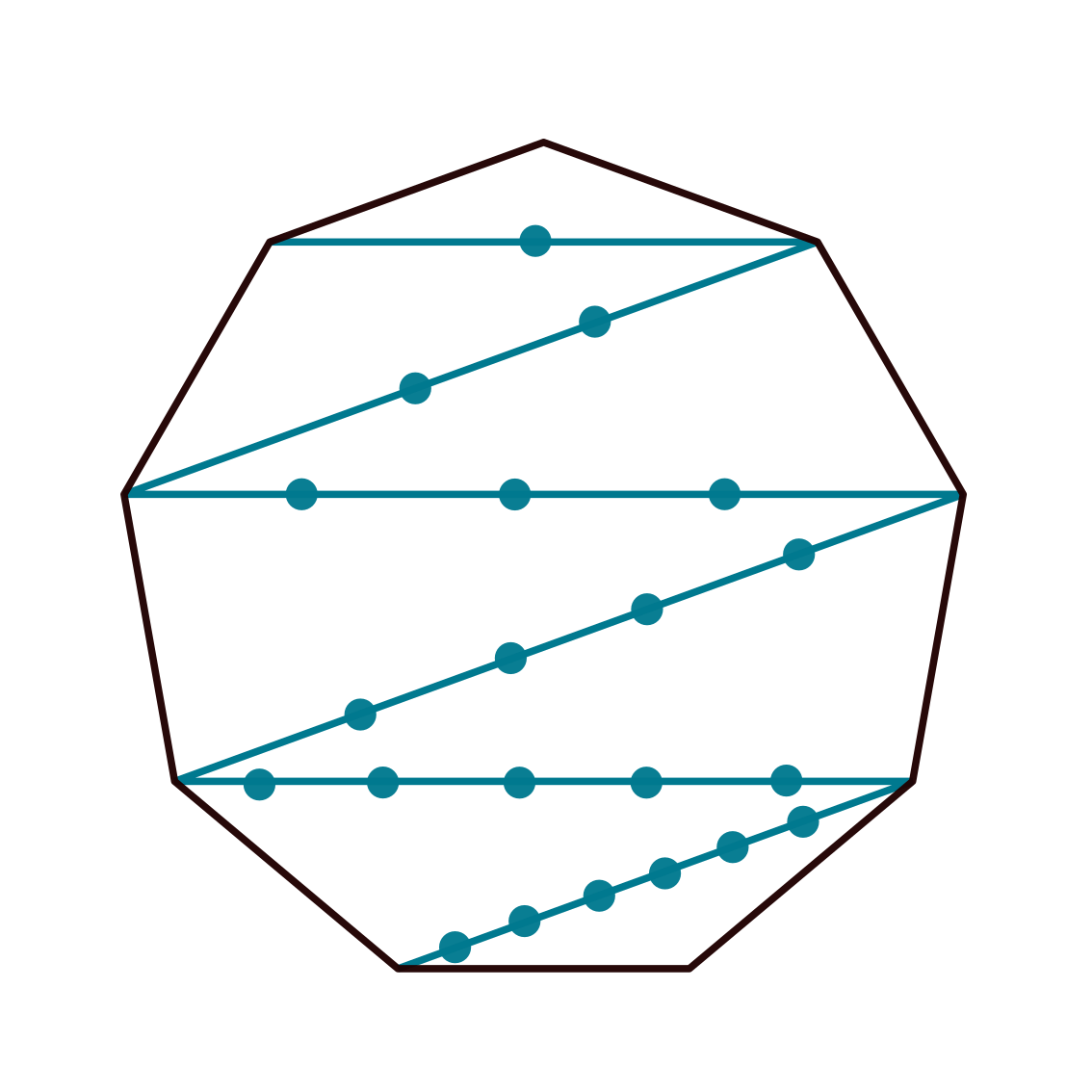}
	\end{subfigure}
	\caption{The transformation of a general graph to a triangulation (after identifying the resulting multiple edges). The dashed edges are such that the coupling constant is set to infinity first, and then to zero (which is equivalent to removing the edges) and hence the spin correlations in the final graph are smaller than in the original graph. }
	\label{f: face degree to 3 general}
\end{figure}

\section{No exponential decay implies a power-law lower bound} \label{sec:Liebsharp}
In this section we finish the proof of the main theorem by showing that the absence of exponential decay implies a power-law lower bound 
on the two-point function when $\Gamma=\mathbb Z^2$. Similar arguments can be applied to other graphs that in addition to being translation invariant possess reflection and rotation symmetries.

\begin{proof}[Proof of Theorem \ref{T:main_BKT}]
Let $0$ denote the vertex at the origin. For a finite subgraph $G$ of $\mathbb Z^2$ containing $0$, let 
\[
\varphi_{G,\beta} =\sum_{w\in \partial G} \langle \sigma_0\bar \sigma_w\rangle_{G,\beta},
\]
where $\partial G$ is the set of vertices of $G$ adjacent to at least one vertex outside $G$.
Define
\begin{align} \label{eq:betac}
\beta_c= \sup \{\beta: \text{there exists finite $G$ with $\varphi_{G,\beta} <1$}  \}.
\end{align}
We will show that $\beta_c$ satisfies the properties listed in Theorem \ref{T:main_BKT}. To this end first fix $\beta<\beta_c$. By Lemma~\ref{L: xy increasing cor}, there exists a finite graph $G$ with $\varphi_{G,\beta} <1$.
Using a standard argument that consists in iteratively applying the Lieb--Rivasseau inequality~\cite{Lieb,Riv} (see Lemma~\ref{L: LR-ineq}) to translates of $G$, we obtain that the two-point functions decay exponentially fast, and hence \emph{(i)} holds true.

To conclude \emph{(ii)}, note that for each finite $G$, $\varphi_{G,\beta}$ is a continuous function of $\beta$, and hence the set in \eqref{eq:betac} is open.
This means that for every $\beta\geq \beta_c$, we have $\varphi_{G,\beta} \geq1$ for all finite subgraphs~$G$. 

\begin{figure}
		\centering
		\includegraphics[scale =0.2]{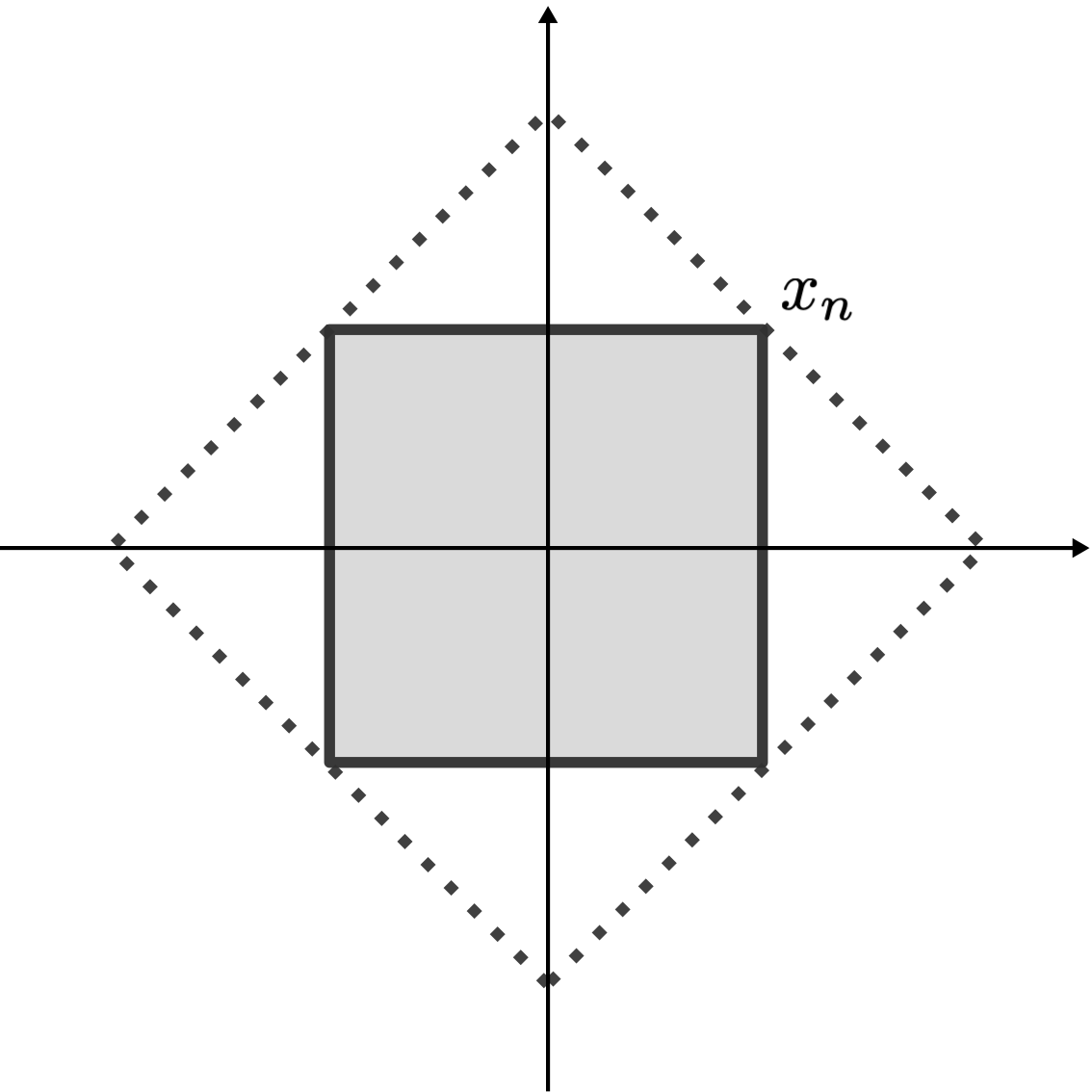}
		\caption{The $[-n, n]^2$ box $\Lambda_n$ shaded in grey and the $L^1$ ball $\Lambda'_n$ of radius $2n$. }
		\label{fig:balls}
\end{figure}

Now let $\Lambda_n$ be the box $[-n, n]^2$, and let $\Lambda'_n$ be the ball in $L^1$ of radius $2n$ (see Figure \ref{fig:balls}). We write $x_n:=(n,n)\in \partial \Lambda_n \cap \partial \Lambda'_n$ and $a_n =\langle \sigma_0\bar \sigma_{x_n}\rangle_{\mathbb Z^2,\beta}$. By rotation symmetry and the Messager--Miracle-Sole~\cite{MMS} inequality (see Lemma~\ref{L: MM-S inequal}), we have
\[
a_n =\min_{v\in \partial \Lambda_n}\langle \sigma_0\bar \sigma_{v}\rangle_{\mathbb Z^2,\beta} = \max_{v\in \partial \Lambda'_n}\langle \sigma_0\bar \sigma_{v}\rangle_{\mathbb Z^2,\beta}.
\]
For $\beta\geq \beta_c$, we moreover have
\[
 \sum_{w \in \partial \Lambda'_n} \langle \sigma_0 \overline{\sigma}_w \rangle_{\mathbb Z^2, \beta} \geq \varphi_{\Lambda_n',\beta} \geq 1.
\] 
These two observations together imply that for any $v\in \partial \Lambda_n$,
	\[
		\langle \sigma_0\bar \sigma_{v}\rangle_{\mathbb Z^2,\beta}  \geq a_n \geq \frac{1}{|\partial \Lambda'_n|} = \frac{1}{8n} \geq  \frac{1}{8 |v|} 
	\]
which implies $(ii)$.
	
Finally by Theorem~\ref{T:BKT planar lattice} we know that there exists a finite $\beta$ at which there is no exponential decay, and by classical expansions there exists a nonzero $\beta$ at which there is exponential decay (see e.g.~\cite{AizSim}). We conclude that $0<\beta_c<\infty$.
\end{proof}

\appendix
\include{appendix_name}
\section{Double currents, path switching and correlation inequalities} \label{sec:doubleswitch}
The main purpose of this section is to present a new technique that may be applied to a further study of the XY model (possibly in higher dimensions).
We develop a loop representation for \emph{squares and products} of correlation functions.
This is a generalization of the construction from Section~\ref{sec:singleswitch}, and to the best of our knowledge has not yet been described in the literature.
It is also analogous to the double random current representation of the Ising model~\cite{GHS,aizenman,ADCS} but is more subtle as one has to deal with path switching rather than connection switching in a percolation model. We stress the fact that we do not use any of the results from this section in the proof of the main theorem, except for the well known inequalities of Lieb and Rivasseau, and Messager and Miracle-Sole.

There will be two major differences in the definition of a loop configuration compared to Section~\ref{sec:singleswitch}: the edges will come in two colours, red and blue, corresponding to two currents $\re$ and $\bl$ respectively,
and we will allow the paths to enter vertices $v$ at which the number of incoming and outgoing edges is not the same, i.e., $\delta (\re+\bl)_v\neq 0$. To be more precise, consider the following definition.
\begin{definition}[Coloured loop configurations outside $S$ with sources $\varphi$]
Let $\mathcal M$ be a multigraph on $V$, let $S\subseteq V$, and $\varphi:V\to \mathbb Z$ with $\sum_{v\in V}\varphi_v=0$. A \emph{coloured loop configuration (on $\mathcal M$) outside $S$ with sources $\varphi$} is 
\begin{itemize} 
\item an assignment of a red or blue color to each edge of $\mathcal M$, together with 
\item a collection of  
\begin{itemize}
\item unrooted directed loops on $\mathcal M$ avoiding $S$, and
\item directed open paths on $ \mathcal M$ not visiting $S$ except possibly at their start and end vertex,
\end{itemize}
such that 
\begin{itemize}
\item every edge of $\mathcal M$ is traversed exactly once by a loop or a path, and
\item at each vertex $v\in V\setminus S$, there are exactly $\varphi_v \id\{ \varphi_v>0\}$ outgoing and $-\varphi_v \id\{ \varphi_v<0\}$ incoming paths.
\end{itemize}
\end{itemize}
We write $\tilde \DCWire^S_{\varphi}$ for the set of all coloured loop configurations outside $S$ with sources $\varphi$, and define a weight on $\tilde \DCWire^S_{\varphi}$ by
\begin{equation} \label{eq: colour weight1}
	\tilde \lambda^S_\beta(\dcwire) = \prod_{v \in V \setminus S} \frac{|\varphi_v|!}{((\deg_{\mathcal M}(v) +|\varphi_v|)/ 2)!}  \prod_{e\in E} \frac{1}{\mathcal M_e!} \Big(\frac{\beta}{2}\Big)^{\mathcal M_e},
\end{equation}
where $\mathcal M$ is the underlying multigraph, and $\mathcal M_e$ is the number of copies of $e$ in $\mathcal M$.  
\end{definition}
Note that this weight no longer only depends on the multigraph $M$ and on $S$, but also on $|\varphi(\omega)|$, where $\varphi(\omega)$ are the sources of $\omega$.
Also note that in the above definition $S$ and $\varphi$ can be chosen independently. $S$ denotes the set of vertices where we do not resolve any connections between paths and loops, and $\varphi$ prescribes where the sources and sinks are (vertices with nonzero value of $\varphi$). At any such vertex $v$, we resolve as many connections as possible leaving only $|\varphi_v|$ incoming or outgoing arrows unmatched, depending on the sign of $\varphi_v$. This is the reason why $\varphi$ appears in the above weight, which was not the case in Section~\ref{sec:singleswitch}.

As before if $S'\subseteq S$, then there is a natural map $\rho: \mathcal L^{S'} \to \mathcal L^S$ that consists in forgetting (or cutting) the loop and path connections at the vertices in $S\setminus S'$, and 
\begin{align}\label{eq:cutting1}
\sum_{\tilde \omega \in \rho^{-1}[\omega]} \tilde \lambda^{S'}_\beta(\tilde \dcwire) =  \tilde \lambda^S_\beta(\dcwire) .
\end{align}

\begin{definition}[Coloured currents and consistent configurations]
We will consider a pair of currents $\re,\bl$ that we think of as \emph{red} and \emph{blue} respectively.
A coloured loop configuration $\omega$ on $\mathcal M_{\re + \bl }$ is called \emph{consistent with $\re$ and $\bl$} if for every edge $vv'\in E$, the number of times the loops and paths traverse a red (resp.\ blue) copy of $vv'$
in the direction of ${(v,v')}$ is equal to 
$\re_{(v,v')}$ (resp.\ $\bl_{(v,v')}$). In particular $\omega$ has sources $\delta(\re+\bl)$.
We define $ \tilde{\DCWire}^S_{\re,\bl}$ to be the set of all coloured loop configurations on $\mathcal M_{\re+\bl }$ outside $S$ that are consistent with $\re$ and $\bl$. 
\end{definition}
For $\varphi,\psi: V\to\mathbb Z$, we also define
\[
  \tilde{\DCWire}_{\varphi,\psi}^S= \bigcup_{\re\in \Omega_{\varphi}, \bl \in \Omega_{\psi}}  \tilde{\DCWire}^S_{\re,\bl}\subseteq   \tilde{\DCWire}_{\varphi+\psi}^S.
\] 
where the union is clearly disjoint. For brevity, we will write $ \tilde \DCWire_{0}^S$ instead of $\tilde \DCWire_{0,0}^S$, where $0$ denotes the zero function on $V$. 

We now relate the weights of loops to those of pairs of currents. To this end, note that for each edge $vv'\in E$, there are exactly 
\[
	\frac{|\re+\bl|_{vv'}!}{\re_{(v, v')}!\re_{{(v', v)}}!\bl_{(v, v')}!\bl_{(v', v)}!}
\]
ways of assigning colour to the copies of $vv'$ in $\mathcal M_{\re + \bl }$, and to orient them in the two possible ways so that the result is consistent with $\re$ and $\bl$.
Moreover, independently of the choices of colours and orientations, there are exactly 
\[
\frac{((\deg_{\mathcal M_{\re + \bl}}(v) +|\varphi_v|)/ 2)!}{|\varphi_v|!}
\] possible pairings of the incoming and outgoing edges at each vertex $v \in V\setminus S$ such that there are exactly $\varphi_v \id\{ \varphi_v>0\}$ outgoing and $-\varphi_v \id\{ \varphi_v<0\}$ incoming edges unpaired. This is equivalent to choosing the possible steps that all the loops and paths in the configuration make
at $v$.
Combining all this, we get the following identity: 
\begin{align} \label{eq:loopexp1}
\sum_{\dcwire \in  \tilde \DCWire^S_{\re,\bl}}\tilde  \lambda^S_\beta(\dcwire) = w_\beta(\re)w_\beta(\bl).
\end{align}
An important observation again is that the right-hand side is independent of $S$, and hence so is the left-hand side.

In particular, for two sourceless currents, we have
\begin{align} \label{eq:squarepartition1}
\sum_{\dcwire \in  \tilde\DCWire^\emptyset_0} \tilde \lambda^{\emptyset}_\beta(\dcwire) = \Big(\sum_{\n \in\Omega_{0}} w_\beta(\n)\Big)^2=(Z^0_{G,\beta})^2.
\end{align}

Again, in the case when $G$ is planar we get the following distributional identity. 
Let $ \tilde {\mathbf{P}}_{G,\beta}$ to be the probability measure on $\tilde {\mathcal L}_0:=\tilde{ \mathcal L}^\emptyset_0$ induced by the weights $\tilde \lambda_{\beta}:=\tilde{ \lambda}^\emptyset_{\beta}$.
For each face $u\in U$ of $G$, and $\omega\in \tilde{ \mathcal L}_0$, define $W_{\omega}(u)$ to be the total net winding
of all the loops in $\omega$ around~$u$. 
\begin{proposition} \label{prop:netwinding1}
The law of $(W(u))_{u\in U}$ under $\tilde {\mathbf{P}}_{G,\beta}$ is the same as the law of the sum of two independent height functions $(h(u)+ h'(u))_{u\in U}$ under ${\mathbb P}_{G,\beta}$.
\end{proposition}

\begin{figure}[h]
	\begin{subfigure}{.2\textwidth}
		\centering	
		\includegraphics[scale =0.12]{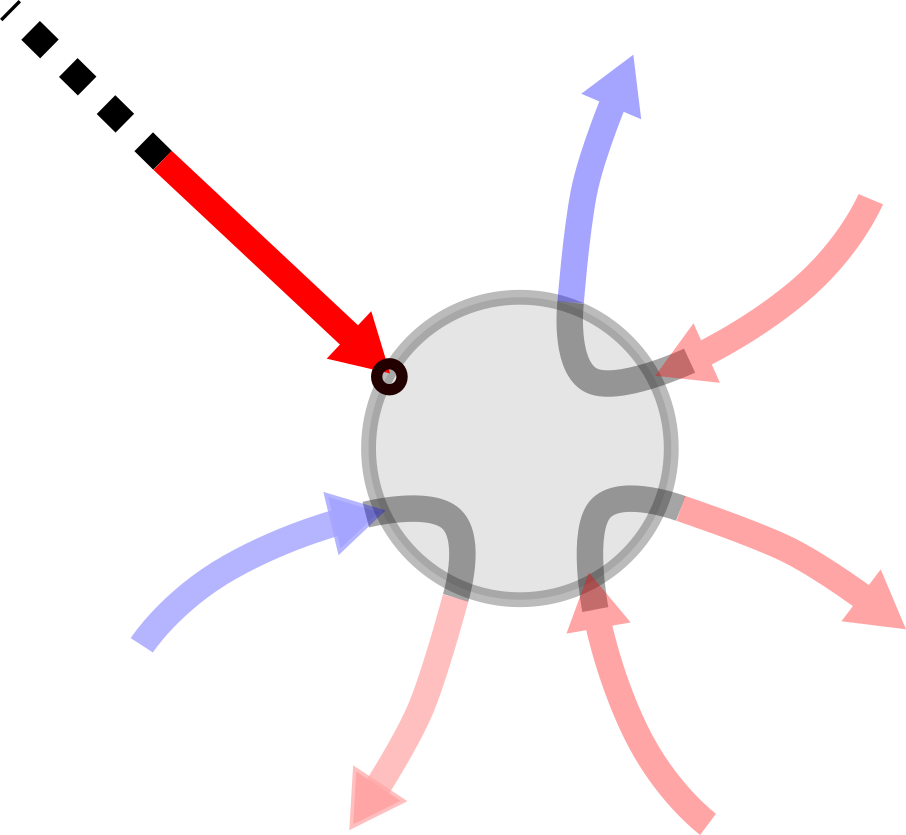}
	\end{subfigure}
	\begin{subfigure}{.2\textwidth}
		\centering
			\includegraphics[scale =0.12]{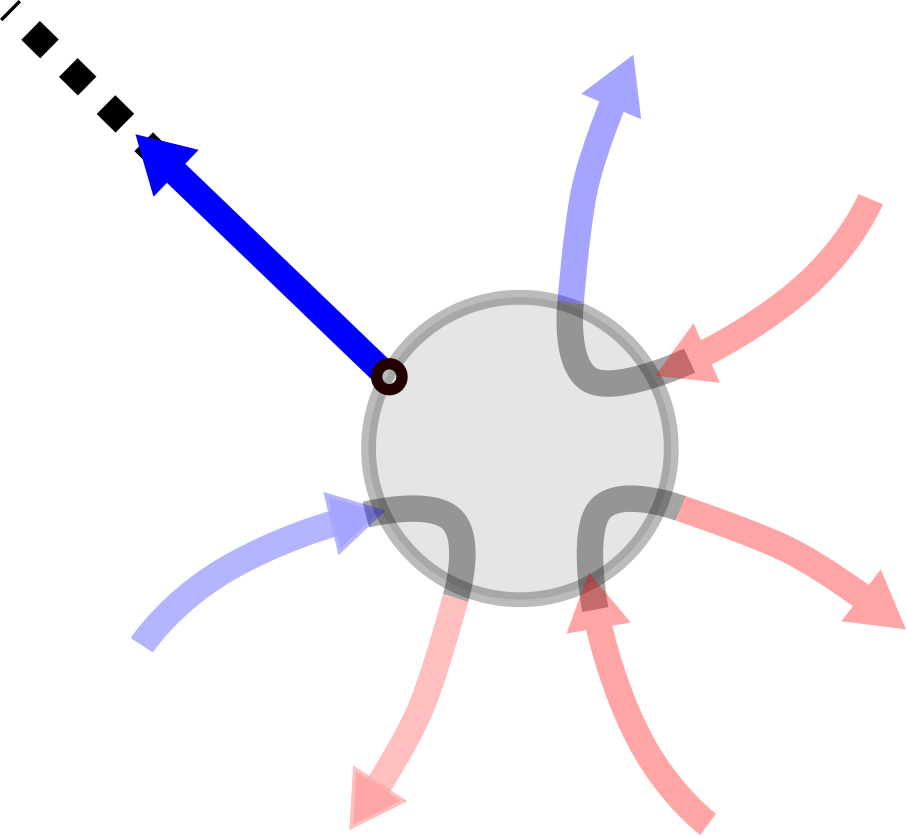}			
	\end{subfigure}
	\begin{subfigure}{.2\textwidth}
		\centering
				\includegraphics[scale =0.12]{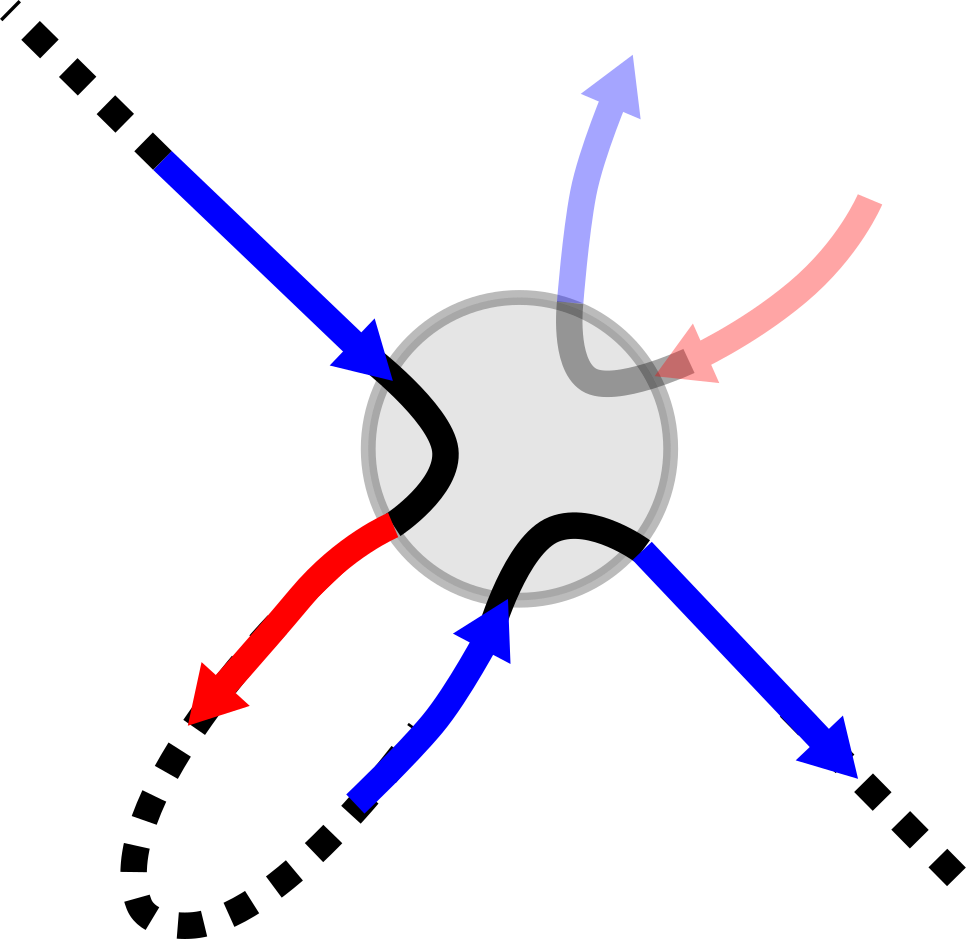}
	\end{subfigure}
	\begin{subfigure}{.2\textwidth}
		\centering
	\includegraphics[scale =0.12]{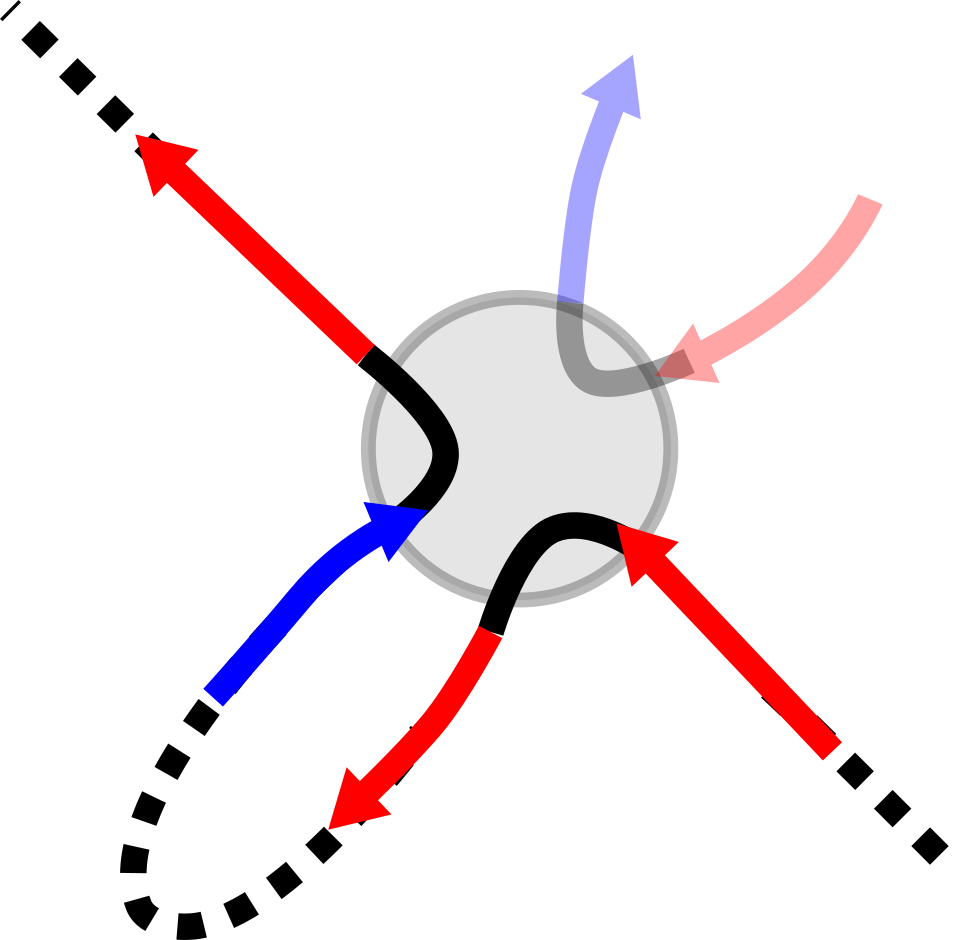}
	\end{subfigure}
	\caption{Path switching behaviour at a vertex $v$ that is (resp.\ is not) the start or end-point of the switched path. Left: the values of \emph{both} $\delta\re_v$ and $\delta\bl_v$ are increased by one after switching. Right: the values are not changed.}
	\label{f: double-switching}
\end{figure}

\subsection{The two point-function and path switching}
We now turn to the loop representation of the square of the two-point function.
To this end, write $\varphi=\delta_a-\delta_b$.
Similar to \eqref{eq:currentexp}, we get 
\begin{align} \label{eq:currentexp1}
\langle \sigma_a \bar \sigma_b \rangle^2_{G,\beta}= \Bigg(\frac{\sum_{\n \in\Omega_{\varphi}} w_\beta(\n) }{Z^0_{G,\beta}}\Bigg)^2=\frac{\sum_{\re,\bl \in\Omega_{\varphi}} w_\beta(\re)w_\beta(\bl) }{(Z^0_{G,\beta})^2}
= \frac{\sum_{ \dcwire \in  \DCWire^{a,b}_{\varphi,\varphi}}  \tilde \lambda^{a,b}_\beta(\dcwire)}{(Z^0_{G,\beta})^2},
\end{align}
where the last equality follows from \eqref{eq:loopexp1}.

As before, we now want to reverse some of the paths. However, this time we also need to take care of the colours of the edges visited by a path. This motivates the following definition.
\begin{definition}[Path switching]
For a path $\gamma$ in a coloured loop configuration $\omega$, we define $s(\gamma)$ to be the path obtained from $\gamma$ by
\begin{itemize}
\item reversing the orientation of $\gamma$, and 
\item swapping the colours of the edges visited by $\gamma$. 
\end{itemize}
We also define $\omega'$ to be the configuration where $\gamma$ is replaced by $s(\gamma)$. 
This operation does not change the underlying multigraph. 
Moreover if $\gamma$ starts at $a$ and ends at $b$, then for any $\varphi,\psi:V\to \mathbb Z$, path switching maps 
\[
\omega \in \tilde{ \mathcal L}^S_{\varphi,\psi}\qquad \textnormal{ to } \qquad \omega'\in \tilde {\mathcal L}^S_{\varphi-\delta_a+\delta_b,\psi-\delta_a+\delta_b}
\]
(see Figure~\ref{f: double-switching}).
\end{definition}
We note that there are two important cases in which path switching does not change the weight $\tilde \lambda^S_\beta$. The first one is when $\{a,b\}\subset S$, and the second one is when
$\varphi_a+\psi_a=1$, and $\varphi_b+\psi_b=-1$, since then the absolute value of the sources of the configuration does not change.

Again the crucial observation now is that switching a path going from $a$ to $b$ maps $\dcwire\in\tilde  \DCWire^{a,b}_{\varphi,\varphi}$ to $ \dcwire' \in \tilde \DCWire^{a,b}_{0}$, and hence erases the sources and sinks of the underlying currents.
Indeed one can easily check (see Figure~\ref{f: double-switching}) that after reversing a path and swapping the colours, the number of incoming minus the number of outgoing red and blue edges at every vertex $v\notin  \{a,b\}$ in $\dcwire'$ is the same as in $\dcwire$, whereas at $a$ and $b$ this number is decreased by one. 
Since we did not change the sources outside $\{a,b\} $, we do not change the weight of a loop configuration, and hence obtain in the same way as in Section~\ref{sec:twopoint} that
\begin{align*}
	\sum_{   \dcwire \in\tilde \DCWire^{a,b}_{\varphi,\varphi}} \tilde \lambda^{a,b}_\beta( \dcwire)={\sum_{   \dcwire \in \tilde \DCWire^\emptyset_0}}\frac{m_{b,a}( \dcwire) }{m_{b,a}( \dcwire)+1 } \tilde \lambda^\emptyset_\beta( \dcwire) \id\{ m_{b,a}( \dcwire) > 0 \}.
\end{align*}
Together with \eqref{eq:currentexp1} this implies the following loop representation of the square of the two-point function. 

\begin{proposition} \label{L:DoubleSwitching}
	Let $a,b\in V$ be distinct. Then
	\[
		\langle \sigma_{a}\bar{\sigma}_{b} \rangle_{G,\beta}^2 =\tilde {\mathbf{E}}_{G,\beta}\Big[\frac{m_{a,b}}{m_{a,b} + 1}\Big], 
	\]
	and in particular
	\[
		\frac{1}{2} \tilde {\mathbf{P}}_{G,\beta}(m_{a,b} > 0) \leq \langle \sigma_{a}\bar{\sigma}_{b} \rangle_{G,\beta}^2 \leq \tilde {\mathbf{P}}_{G,\beta}(m_{a,b} >0). 
	\]
\end{proposition}

\begin{remark}\label{rem:squares}
The constant in the inequality of Lemma~\ref{lem:squares} can be improved to $1$ using the same method as above but starting from coloured loop configurations 
in $\tilde {\mathcal L}^{a,b}_{0,2\varphi}$ instead of $\tilde {\mathcal L}^{a,b}_{\varphi,\varphi}$.
\end{remark}

\subsection{Application to some inequalities}
As a further application we now prove an inequality that is related to but independent of the Ginibre inequality.

\begin{lemma} \label{L: ferro-magnet inequal} Let $a,b,c\in V$. Then
	\[
	\langle \sigma_{a}\bar{\sigma}_{b} \rangle_{G,\beta} \geq \langle \sigma_{a}\bar{\sigma}_{c} \rangle_{G,\beta}\langle \sigma_{c}\bar{\sigma}_{b} \rangle_{G,\beta} \geq   \langle \sigma_{a}{\sigma}_{b}\bar{\sigma}^2_{c}\rangle_{G,\beta}.
	\]
\end{lemma}

\begin{proof}
The two inequalities have, maybe quite surprisingly, almost the same proof. We only prove the first and leave the second to the reader. We set $S = \{ c\}$ and will write $c$ instead of $\{c\}$ in our notation. We also define $\varphi = \delta_a - \delta_c$, $\psi = \delta_b - \delta_c$,
and note that $\psi -\varphi=\delta_b - \delta_a$. Also note that for each $\dcwire \in \tilde{\DCWire}^{c}_{\varphi, \psi}$, the unique path starting at $a$ must end at $c$. Consider the map $\omega \mapsto \omega'$ that switches this path. 
Clearly this is a bijection between $\tilde{\DCWire}^c_{\varphi, \psi}$ and 
\[
\{\dcwire \in \tilde{\DCWire}^{c}_{0, \psi-\varphi}: \text{the unique path ending at $a$ starts at $c$}\}.
\] 
Moreover, we have $|\varphi_v(\omega)|=|\varphi_v(\omega')|$ for all $v\neq c$, and hence the weights $\tilde \lambda^c_\beta$ are preserved.
This means that
\begin{align*}
	\langle \bar{\sigma}_{a} \sigma_c \rangle_{G, \beta} \langle \bar{\sigma}_b \sigma_c \rangle_{G, \beta} (Z_{G, \beta}^0)^2 = \sum_{\dcwire \in \tilde{\DCWire}^{c}_{\varphi, \psi}} \tilde{\lambda}_{\beta}^{c}(\dcwire) = \sum_{\dcwire \in \tilde{\DCWire}_{0, \psi-\varphi}^{c}} \id{\{c \to a \text{ in } \omega\}} \tilde{\lambda}_{\beta}^{c}(\dcwire) \leq \langle \sigma_a \bar{\sigma}_b \rangle_{G, \beta} (Z_{G, \beta}^{0})^2,
\end{align*}
where we used \eqref{eq:loopexp1} twice. This finishes the proof.
\end{proof}

\begin{remark}
	Note that the Ginibre inequality in e.g.\ \cite[Theorem 2.3]{PelSpi} is equivalent to a relation between increments given by
	\[
		\langle \sigma_a \bar{\sigma}_b \rangle_{G, \beta} - \langle \sigma_a \bar{\sigma}_c \rangle_{G, \beta} \langle \sigma_{c}\bar{\sigma}_b \rangle_{G, \beta} \geq \langle \sigma_a \bar{\sigma}_c \rangle_{G, \beta} \langle \sigma_{c}\bar{\sigma}_b \rangle_{G, \beta} - \langle \sigma_a \sigma_b \bar{\sigma}_c^2 \rangle_{G,\beta}.
	\]
	Comparing this with the statement of Lemma \ref{L: ferro-magnet inequal}, the latter proves nonnegativity of the increments. Hence, the Ginibre inequality and Lemma \ref{L: ferro-magnet inequal} do not imply one another. 
\end{remark}

The purpose of the remainder of this section is to give more applications of the representation introduced above.
We start with two new bijective proofs of the classical inequalities that we used in the proof of our main theorem.
\begin{lemma}[Lieb--Rivasseau inequality~\cite{Lieb,Riv}] \label{L: LR-ineq}
	Let $G=(V,E)$ be any graph. Let $a,b \in V$ be distinct, and let $H$ be a finite subgraph of $G$ containing $a$ and not containing $b$, and let $\partial H$ be the set of vertices of $H$ adjacent to at least one vertex outside $H$. Then
	\[
		\langle \sigma_a \bar {\sigma}_{b} \rangle_{G,\beta} \leq \sum_{c \in \partial H} \langle \sigma_a \bar{\sigma}_c \rangle_{H,\beta} \langle \sigma_c  \bar{\sigma}_{b}\rangle_{G,\beta}.
	\]
\end{lemma}
\begin{proof}
It is enough to assume that $G$ is finite, and then approximate an infinite graph by finite subgraphs. The proof is similar to the previous one. Assume $a\notin \partial H$. Otherwise, there is nothing to prove. Fix $c \in \partial H$ and $S = \{ c\}$. We will write $c$ instead of $\{c\}$ in our notation. Let $\varphi= \delta_c - \delta_a$, $\psi= \delta_c - \delta_b$, and note that $\psi - \varphi=\delta_a - \delta_b$. 

Write $\tilde{\mathcal  L}_c$ for the collection of coloured loop configurations $\dcwire \in \tilde \DCWire^{c}_{0, \psi-\varphi}$ with the property that the unique path starting at $a$ exits $H\setminus \partial H$ at $c$, and $\omega$ has no red edges outside of $H$.
For $\omega\in \tilde {\mathcal L}_c$ consider a coloured loop configuration where this path is switched. Clearly this is a bijection between $\tilde {\mathcal L}_c$ and the set of configurations 
 $\omega'\in \tilde \DCWire^{c}_{\varphi, \psi}$ that have no red edges outside $H$, and for which the unique path ending at $a$ stays within $H\setminus \partial H$ until it hits $c$. Denote this collection of configurations by $\tilde {\mathcal L}'_c$. Moreover, we have $|\varphi_v(\omega)|=|\varphi_v(\omega')|$ for all $v\neq c$, and hence the weights $\tilde \lambda^c_\beta$ are preserved.

Let $\tilde{\c{E}}_c$ be the collection of $\dcwire \in \tilde{\DCWire}_{0, \psi-\varphi}$ with the property that the unique path from $a$ to $b$ exits $H \setminus \partial H$ in $c$, and $\omega$ does not have red edges outside of $H$. Clearly, the subset of $\tilde \DCWire_{0, \psi-\varphi}$ consisting of configurations with no red edges outside of $H$ equals the disjoint union $\cup_{c \in \partial H} \tilde{\c{E}}_c$ and cutting $\omega \in \tilde{\c{E}}_c$ at $c$ gives an element of $\tilde{\c{L}}_c$. In light of \eqref{eq:cutting1}, we therefore have
	\begin{align*}
		\langle \sigma_a \bar{\sigma}_{b} \rangle_{G, \beta} Z^0_{G, \beta} Z^0_{H, \beta} = \sum_{c \in \partial H} \sum_{\dcwire \in\tilde{\mathcal L}_c} \tilde \lambda^{c}_{\beta}(\dcwire) 
		= \sum_{c \in \partial H} \sum_{\dcwire' \in \tilde{\mathcal L}_c'}   \tilde \lambda^{c}_{\beta}(\dcwire') 
		\leq \sum_{c \in \partial H} \langle \sigma_c \bar{\sigma}_a \rangle_{H, \beta} \langle \sigma_c \bar{\sigma}_b \rangle_{G, \beta}Z^0_{G, \beta} Z^0_{H, \beta},
	\end{align*}
	which completes the proof. 		
\end{proof}

We are also able to use the coloured loop representation to prove the Messager--Miracle-Sole inequality. 
\begin{lemma}[Messager--Miracle-Sole inequality~\cite{MMS}] \label{L: MM-S inequal}
For any $n\in \mathbb Z$, the two sequences $\langle \sigma_0 \bar{\sigma}_{(n,k)} \rangle_{\ZZ^2, \beta}$ and $\langle \sigma_0 \bar{\sigma}_{(n+k,n-k)} \rangle_{\ZZ^2, \beta}$  
are nonincreasing in $k$ for $k\geq 0$.
\end{lemma}

Geometrically, this in particular implies that the largest correlation with the spin at $0$ on any vertical, horizontal or diagonal straight line is attained by the vertex closest to $0$.
This will follow from the following lemma after taking $G\nearrow \mathbb Z^2$. The proof is inspired by the one from~\cite{ADTW} for the Ising model.
The idea is to fold a graph across a line and think of the parts of the current coming from both sides of the line as the red and blue current in the coloured loop representation.
\begin{lemma}
Let $G=(V,E)$ be a subgraph of $\mathbb Z^2$ symmetric under reflection across a line $L$. Let $a,b\in V$ lie on the same side of $L$, and let $L(b)\in V$ be the reflection of $b$. Then
\[
\langle \sigma_a\bar\sigma_b\rangle_{G,\beta} \geq \langle \sigma_a\bar\sigma_{L(b)}\rangle_{G,\beta}.
\]
\end{lemma}

\begin{proof}
We only consider the easier case when $L$ passes through vertices. 
This means that it is either a diagonal, or a horizontal (vertical) line at integer height.
The more involved case when $L$ passes only through the edges (this case implies Lemma~\ref{L: MM-S inequal} for horizontal and vertical lines) we leave to the interested reader.

If $L$ is horizontal or vertical, then split the edges that lie on $L$ into two parallel 
edges with coupling constants $\beta/2$, and think of the resulting graph as a new graph $G$.
Write $Z$ for the set of vertices on $L$, and $G_-=(V_-,E_-)$ and $G_+=(V_+,E_+)$ for the two isomorphic parts of $G$ separated by $L$ where $G_-$ contains $a$ and $b$ (each of them also containing $Z$).

We can decompose a current $\n$ on $G$ into two parts: $\re$ and $\bl$ on $G_-$ and $G_+$. 
In what follows, we identify $G_-$ with $G_+$ under the obvious isomorphism, and all currents are considered on $G_-$ unless stated otherwise.
Let $\c{C}_k$, {for} $k=0,1$, be the set of functions $\varphi: V_- \to \mathbb Z$ such that $\varphi_v=0$ for $v\in V_-\setminus Z$, and $\sum_{v\in Z}\varphi_v=k$.
Since every current in $\Omega_{\delta_a-\delta_{L(b)}}(G)$ must have a total flux of $+1$ across $L$, we can write
	\begin{align*}
		\langle \sigma_a \bar{\sigma}_{L(b)} \rangle_{G, \beta}Z_{G,\beta}^0 &= 
		\sum_{\varphi \in \c{C}_1} \sum_{\re \in \Omega_{\delta_a-\varphi}, \bl \in \Omega_{-\delta_{b}+\varphi}} w_\beta(\re)w_\beta(\bl) \\
		&= \sum_{\varphi \in \c{C}_1} \sum_{\re \in \Omega_{\delta_a-\varphi}, \bl \in \Omega_{\delta_{b}-\varphi}} w_\beta(\re)w_\beta(\bl) \\
		 &= \sum_{\varphi \in \c{C}_1}  \sum_{\omega\in \tilde{\mathcal L}^Z_{\delta_a-\varphi,\delta_{b}-\varphi} }\tilde \lambda^Z_\beta(\omega) ,
	\end{align*}
	where the second inequality holds true as a the weight $w_{\beta}$ is invariant under reversal of the current, and the last equality is a consequence of~\eqref{eq:loopexp1}.
Now, for each $\omega\in \tilde{\mathcal L}^Z_{\delta_0-\varphi,\delta_{b}-\varphi}$ switch the unique path $\gamma$ starting at $b$. This transformation preserves weights and results in a configuration 
$\omega'\in \tilde{\mathcal L}^Z_{\delta_0-\delta_{b}-\varphi',-\varphi'}$, where $\varphi'=\varphi-\delta_z\in \c{C}_0$ and $z\in Z$ is the vertex at which $\gamma$ ends.
Reversing the order of the steps above we therefore get
\begin{align*}
\langle \sigma_a \bar{\sigma}_{L(b)} \rangle_{G, \beta}Z_{G,\beta}^0& =  \sum_{\varphi \in \c{C}_1} \sum_{z\in Z}\sum_{\omega'\in \tilde{\mathcal L}^Z_{\delta_a-\delta_{b}-\varphi',-\varphi'} }\tilde \lambda^Z_\beta(\omega') \id\{\text{the path ending at $b$ starts at $z$} \} \\
&=  \sum_{\varphi' \in \c{C}_0} \sum_{\omega'\in \tilde{\mathcal L}^Z_{\delta_a-\delta_{b}-\varphi',-\varphi'} }\tilde \lambda^Z_\beta(\omega') \id\{\text{the path ending at $b$ starts in $Z$} \} \\
&\leq  \sum_{\varphi' \in \c{C}_0} \sum_{\re \in \Omega_{\delta_a-\delta_b-\varphi'}, \bl \in \Omega_{\varphi'}} w_\beta(\re)w_\beta(\bl)\\
&=\langle \sigma_a \bar{\sigma}_b \rangle_{G, \beta}Z_{G,\beta}^0,
\end{align*}
where the last equality follows since the total flux of a current in $\Omega_{\delta_a-\delta_b}(G)$ across $L$ is zero.
\end{proof}

\subsection{Limitations of the coloured loop representation}
A natural idea is to try to prove the Ginibre inequality in form of Lemma~\ref{L: xy increasing cor} using our representation. One would like to show that the derivative of the two-point function with respect to one coupling constant $J_e$ is nonnegative.
Using coloured loop configurations we can write
\begin{align*}
(Z^{0}_{G,\beta})^2\frac{\partial  }{\partial J_e}\langle\sigma_a\bar\sigma_b \rangle_{G,\beta}&= Z^{0}_{G,\beta}{\frac{\partial }{\partial J_e}Z^{\delta_a-\delta_b}_{G,\beta}}-Z^{\delta_a-\delta_b}_{G,\beta}{\frac{\partial }{\partial J_e}Z^{0}_{G,\beta}} \\
&=J^{-1}_e\sum_{\omega \in \tilde{\mathcal L}^{\emptyset}_{\delta_a-\delta_b,0}} \tilde \lambda^{\emptyset}_{\beta}(\omega)(R_e(\omega)-B_e(\omega)),
\end{align*}
where $R_e(\omega)$ and $B_e(\omega)$ is respectively the number of red and blue copies of $e$ in the multigraph visited by the unique path from $a$ to $b$ in $\omega$.
Without going into too many details, to justify the second equality we make the following observations. First, taking the derivative with respect to $J_e$ is equivalent to dividing by $J_e$ and marking one of the 
copies of $e$ of the right colour (here the currents in $\Omega_{\delta_a-\delta_b}$ are red and those in $\Omega_0$ are blue). Then, if the marked edge is not on the path from $a$ to $b$, we switch the corresponding loop (reverse it and swap the colours). This does not change the weight of the configuration. Such terms hence cancel out from the expression above as the loops going trough a marked blue copy of $e$ are counted with a minus sign.
The remaining terms are those whose marked edge lies on the distinguished path. This gives the final formula.

Clearly, the final result is not evidently nonnegative and we would need additional arguments to conclude the Ginibre inequality. On the other hand, the Ginibre inequality implies the distinguished path visits red edges more often than blue edges on average.

\bibliographystyle{amsplain}
\bibliography{XYHF}

% \bib, bibdiv, biblist are defined by the amsrefs package.
\begin{bibdiv}
\begin{biblist}

\bib{aizenman}{article}{
      author={Aizenman, M.},
       title={{Geometric analysis of $\varphi^{4}$ fields and Ising models. I,
  II}},
        date={1982},
     journal={Comm. Math. Phys.},
      volume={86},
      number={1},
       pages={1\ndash 48},
         url={http://projecteuclid.org/euclid.cmp/1103921614},
}

\bib{AizSim}{article}{
      author={Aizenman, M.},
      author={Simon, B.},
       title={{A comparison of plane rotor and Ising models}},
        date={1980},
        ISSN={0375-9601},
     journal={Physics Letters A},
      volume={76},
      number={3},
       pages={281\ndash 282},
  url={https://www.sciencedirect.com/science/article/pii/0375960180904934},
}

\bib{ADCS}{article}{
      author={Aizenman, Michael},
      author={Duminil-Copin, Hugo},
      author={Sidoravicius, Vladas},
       title={{Random Currents and Continuity of Ising Model's Spontaneous
  Magnetization}},
        date={2015},
        ISSN={1432-0916},
     journal={Communications in Mathematical Physics},
      volume={334},
      number={2},
       pages={719\ndash 742},
         url={http://dx.doi.org/10.1007/s00220-014-2093-y},
}

\bib{ADTW}{article}{
      author={Aizenman, Michael},
      author={Duminil-Copin, Hugo},
      author={Tassion, Vincent},
      author={Warzel, Simone},
       title={{Emergent planarity in two-dimensional Ising models with
  finite-range Interactions}},
        date={2019},
     journal={Inventiones mathematicae},
      volume={216},
      number={3},
       pages={661\ndash 743},
         url={https://doi.org/10.1007/s00222-018-00851-4},
}

\bib{BLU}{book}{
      author={Benassi, Costanza},
      author={Lees, Benjamin},
      author={Ueltschi, Daniel},
      editor={Michelangeli, Alessandro},
      editor={Dell'Antonio, Gianfausto},
       title={{Correlation Inequalities for Classical and Quantum XY Models}},
   publisher={Springer International Publishing},
        date={2017},
        ISBN={978-3-319-58904-6},
         url={https://doi.org/10.1007/978-3-319-58904-6_2},
}

\bib{BenUel}{article}{
      author={Benassi, Costanza},
      author={Ueltschi, Daniel},
       title={Loop correlations in random wire models},
        date={2020},
     journal={Communications in Mathematical Physics},
      volume={374},
      number={2},
       pages={525\ndash 547},
}

\bib{Ber1}{article}{
      author={Berezinskii, V.~L.},
       title={{Destruction of long range order in one-dimensional and
  two-dimensional systems having a continuous symmetry group. I. Classical
  systems}},
        date={1971},
     journal={Sov. Phys. JETP},
      volume={32},
       pages={493\ndash 500},
}

\bib{Ber2}{article}{
      author={Berezinskii, V.~L.},
       title={{Destruction of Long-range Order in One-dimensional and
  Two-dimensional Systems Possessing a Continuous Symmetry Group. II. Quantum
  Systems.}},
        date={1972},
     journal={Sov. Phys. JETP},
      volume={34},
      number={3},
       pages={610},
}

\bib{BFS}{article}{
      author={Brydges, D.},
      author={Fr\"{o}hlich, J.},
      author={Spencer, T.},
       title={{The random walk representation of classical spin systems and
  correlation inequalities}},
        date={1982},
     journal={Communications in Mathematical Physics},
      volume={83},
      number={1},
       pages={123 \ndash  150},
         url={https://doi.org/},
}

\bib{CPST}{article}{
      author={Chandgotia, N.},
      author={Peled, R.},
      author={Sheffield, S.},
      author={Tassy, M.},
       title={{Delocalization of uniform graph homomorphisms from
  $\mathbb{Z}^2$ to $\mathbb{Z}$}},
        date={2021},
     journal={Communications in Mathematical Physics},
}

\bib{dubedat2011topics}{article}{
      author={Dub{\'e}dat, Julien},
       title={Topics on abelian spin models and related problems},
        date={2011},
     journal={Probability Surveys},
      volume={8},
       pages={374\ndash 402},
}

\bib{DCL}{article}{
      author={Duminil-Copin, H.},
      author={Lis, M.},
       title={On the double random current nesting field},
        date={2019},
     journal={Probability Theory and Related Fields},
      volume={175},
      number={3},
       pages={937\ndash 955},
         url={https://doi.org/10.1007/s00440-019-00899-0},
}

\bib{DCT}{article}{
      author={Duminil-Copin, Hugo},
      author={Tassion, Vincent},
       title={{A New Proof of the Sharpness of the Phase Transition for
  Bernoulli Percolation and the Ising Model}},
        date={2016},
     journal={Communications in Mathematical Physics},
      volume={343},
      number={2},
       pages={725\ndash 745},
         url={https://doi.org/10.1007/s00220-015-2480-z},
}

\bib{FriVel}{book}{
      author={Friedli, S.},
      author={Velenik, Y.},
       title={{Statistical Mechanics of Lattice Systems: A Concrete
  Mathematical Introduction}},
   publisher={Cambridge University Press},
        date={2017},
        ISBN={978-1-107-18482-4},
}

\bib{FSS}{article}{
      author={Fr{\"o}hlich, J.},
      author={Simon, B.},
      author={Spencer, T.},
       title={Infrared bounds, phase transitions and continuous symmetry
  breaking},
        date={1976},
     journal={Communications in Mathematical Physics},
      volume={50},
      number={1},
       pages={79\ndash 95},
         url={https://doi.org/10.1007/BF01608557},
}

\bib{FroSpe}{article}{
      author={Fr{\"o}hlich, J.},
      author={Spencer, T.},
       title={{The Kosterlitz-Thouless transition in two-dimensional Abelian
  spin systems and the Coulomb gas}},
        date={1981},
     journal={Communications in Mathematical Physics},
      volume={81},
      number={4},
       pages={527\ndash 602},
         url={https://doi.org/10.1007/BF01208273},
}

\bib{GS}{article}{
      author={Garban, Christophe},
      author={Sep{\'u}lveda, Avelio},
       title={{Statistical reconstruction of the Gaussian free field and KT
  transition}},
        date={2020},
     journal={arXiv preprint arXiv:2002.12284},
}

\bib{Gin}{article}{
      author={Ginibre, J.},
       title={{General formulation of Griffiths' inequalities}},
        date={1970},
     journal={Communications in Mathematical Physics},
      volume={16},
      number={4},
       pages={310 \ndash  328},
         url={https://doi.org/},
}

\bib{GHS}{article}{
      author={Griffiths, R.~B.},
      author={Hurst, C.~A.},
      author={Sherman, S.},
       title={{Concavity of Magnetization of an Ising Ferromagnet in a Positive
  External Field}},
        date={1970},
     journal={Journal of Mathematical Physics},
      volume={11},
      number={3},
       pages={790\ndash 795},
  url={http://scitation.aip.org/content/aip/journal/jmp/11/3/10.1063/1.1665211},
}

\bib{KP}{article}{
      author={Kharash, Vital},
      author={Peled, Ron},
       title={{The Fr\"{o}hlich-Spencer Proof of the
  Berezinskii-Kosterlitz-Thouless Transition}},
        date={2017},
     journal={arXiv preprint arXiv:1711.04720},
}

\bib{KosTho}{article}{
      author={Kosterlitz, J.~M.},
      author={Thouless, D.~J.},
       title={Ordering, metastability and phase transitions in two-dimensional
  systems},
        date={1973apr},
     journal={Journal of Physics C: Solid State Physics},
      volume={6},
      number={7},
       pages={1181\ndash 1203},
         url={https://doi.org/10.1088/0022-3719/6/7/010},
}

\bib{Lammers}{article}{
      author={Lammers, P.},
       title={Height function delocalisation on cubic planar graphs},
        date={2021},
     journal={Probability Theory and Related Fields},
         url={https://doi.org/10.1007/s00440-021-01087-9},
}

\bib{LamOtt}{article}{
      author={Lammers, Piet},
      author={Ott, S{\'e}bastien},
       title={{Delocalisation and absolute-value-FKG in the solid-on-solid
  model}},
        date={2021},
     journal={arXiv preprint arXiv:2101.05139},
}

\bib{LeJan}{article}{
      author={Le~Jan, Y.},
       title={Markov loops and renormalization},
        date={2010/5/1},
     journal={The Annals of Probability},
      volume={38},
      number={3},
       pages={1280\ndash 1319},
         url={https://doi.org/10.1214/09-AOP509},
}

\bib{LeJan1}{article}{
      author={Le~Jan, Y.},
       title={{Markov loops, free field and Eulerian networks}},
        date={2015},
     journal={Journal of the Mathematical Society of Japan},
      volume={67},
      number={4},
       pages={1671 \ndash  1680},
         url={https://doi.org/10.2969/jmsj/06741671},
}

\bib{LeeTag}{article}{
      author={Lees, Benjamin},
      author={Taggi, Lorenzo},
       title={{Exponential decay of transverse correlations for O(N) spin
  systems and related models}},
        date={2021},
     journal={Probability Theory and Related Fields},
      volume={180},
      number={3},
       pages={1099\ndash 1133},
         url={https://doi.org/10.1007/s00440-021-01053-5},
}

\bib{Lieb}{article}{
      author={Lieb, Elliott~H.},
       title={{A refinement of Simon's correlation inequality}},
        date={1980},
     journal={Comm. Math. Phys.},
      volume={77},
      number={2},
       pages={127\ndash 135},
         url={https://projecteuclid.org:443/euclid.cmp/1103908400},
}

\bib{Lis19}{unpublished}{
      author={Lis, M.},
       title={Spins, percolation and height functions},
        date={2019},
        note={arXiv:1909.07351},
}

\bib{LisAT}{article}{
      author={Lis, Marcin},
       title={{On Boundary Correlations in Planar Ashkin--Teller Models}},
        date={202102},
        ISSN={1073-7928},
     journal={International Mathematics Research Notices},
         url={https://doi.org/10.1093/imrn/rnaa380},
        note={rnaa380},
}

\bib{Lupu}{article}{
      author={Lupu, Titus},
       title={{From loop clusters and random interlacements to the free
  field}},
        date={2016},
     journal={The Annals of Probability},
      volume={44},
      number={3},
       pages={2117 \ndash  2146},
         url={https://doi.org/10.1214/15-AOP1019},
}

\bib{MerWag}{article}{
      author={Mermin, N.~D.},
      author={Wagner, H.},
       title={{Absence of Ferromagnetism or Antiferromagnetism in One- or
  Two-Dimensional Isotropic Heisenberg Models}},
        date={1966Nov},
     journal={Phys. Rev. Lett.},
      volume={17},
       pages={1133\ndash 1136},
         url={https://link.aps.org/doi/10.1103/PhysRevLett.17.1133},
}

\bib{MMS}{article}{
      author={Messager, A.},
      author={Miracle-Sole, S.},
       title={{Correlation functions and boundary conditions in the Ising
  ferromagnet}},
        date={1977},
     journal={Journal of Statistical Physics},
      volume={17},
      number={4},
       pages={245\ndash 262},
         url={https://doi.org/10.1007/BF01040105},
}

\bib{PelSpi}{incollection}{
      author={Peled, Ron},
      author={Spinka, Yinon},
       title={{Lectures on the spin and loop $O (n)$ models}},
        date={2019},
   booktitle={Sojourns in probability theory and statistical physics-i},
   publisher={Springer},
       pages={246\ndash 320},
}

\bib{Riv}{article}{
      author={Rivasseau, V.},
       title={{Lieb's correlation inequality for plane rotors}},
        date={1980},
     journal={Communications in Mathematical Physics},
      volume={77},
      number={2},
       pages={145 \ndash  147},
         url={https://doi.org/},
}

\bib{Segura}{article}{
      author={Segura, Javier},
       title={{Bounds for ratios of modified Bessel functions and associated
  Tur{\'a}n-type inequalities}},
        date={2011},
     journal={Journal of Mathematical Analysis and Applications},
      volume={374},
      number={2},
       pages={516\ndash 528},
}

\bib{sheffield}{book}{
      author={Sheffield, Scott},
       title={Random surfaces},
    language={en},
      series={Ast\'erisque},
   publisher={Soci\'et\'e math\'ematique de France},
        date={2005},
      number={304},
         url={http://www.numdam.org/item/AST_2005__304__R1_0/},
      review={\MR{2251117}},
}

\bib{Sym}{techreport}{
      author={Symanzik, K.},
       title={Euclidean quantum field theory},
 institution={New York Univ., NY},
        date={1969},
}

\bib{Bessel}{inproceedings}{
      author={Thiruvenkatachar, V.~R.},
      author={Nanjundiah, T.~S.},
       title={{Inequalities concerning Bessel functions and orthogonal
  polynomials}},
organization={Springer},
        date={1951},
   booktitle={Proceedings of the indian academy of sciences-section a},
      volume={33},
       pages={373},
}

\end{biblist}
\end{bibdiv}
\end{document}